\documentclass[oneside,english]{amsart}
\usepackage[T1]{fontenc}
\usepackage[latin9]{inputenc}
\usepackage{textcomp}
\usepackage{amsbsy}
\usepackage{amstext}
\usepackage{amsthm}
\usepackage{esint}

\makeatletter
\numberwithin{equation}{section}
\numberwithin{figure}{section}
\theoremstyle{plain}
\newtheorem{thm}{\protect\theoremname}[section]
  \theoremstyle{remark}
  \newtheorem{rem}[thm]{\protect\remarkname}
  \theoremstyle{plain}
  \newtheorem{lem}[thm]{\protect\lemmaname}
  \theoremstyle{plain}
  \newtheorem{prop}[thm]{\protect\propositionname}
  \theoremstyle{definition}
  \newtheorem{defn}[thm]{\protect\definitionname}
  \theoremstyle{definition}
  \newtheorem{example}[thm]{\protect\examplename}

\def\makebbb#1{
    \expandafter\gdef\csname#1\endcsname{
        \ensuremath{\Bbb{#1}}}
}\makebbb{R}\makebbb{N}\makebbb{Z}\makebbb{C}\makebbb{H}\makebbb{E}\makebbb{H}\makebbb{P}\makebbb{B}\makebbb{Q}\makebbb{E}

\usepackage{babel}

\usepackage{babel}

\makeatother

\usepackage{babel}

\makeatother

\usepackage{babel}
  \providecommand{\definitionname}{Definition}
  \providecommand{\examplename}{Example}
  \providecommand{\lemmaname}{Lemma}
  \providecommand{\propositionname}{Proposition}
  \providecommand{\remarkname}{Remark}
\providecommand{\theoremname}{Theorem}

\begin{document}

\title[Propagation of Chaos for a class of singular first order models]{Propagation of chaos for a class of first order models with singular
mean field interactions}

\author{Robert J. Berman, Magnus Önnheim}

\email{robertb@chalmers.se, onnheimm@chalmers.se}

\address{Department of Mathematical Sciences, Chalmers University of Technology
and University of Gothenburg, 412 96 Göteborg, Sweden}
\begin{abstract}
Dynamical systems of $N$ particles in {\normalsize{}$\R^{D}$} interacting
by a singular pair potential of mean field type are considered. The
systems are assumed to be of gradient type and the existence of a
macroscopic limit in the many particle limit is established for a
large class of singular interaction potentials in the stochastic as
well as the deterministic settings. The main assumption on the potentials
is an appropriate notion of quasi-convexity. When $D=1$ the convergence
result is sharp when applied to strongly singular repulsive interactions
and for a general dimension $D$ the result applies to attractive
interactions with Lipschitz singular interaction potentials, leading
to stochastic particle solutions to the corresponding macroscopic
aggregation equations. The proof uses the theory of gradient flows
in Wasserstein spaces of Ambrosio-Gigli-Savaree.
\end{abstract}

\maketitle

\section{Introduction}

Let $F$ be an odd map $\R^{D}\rightarrow\R^{D}$ and consider the
following system of $N$ stochastic differential equations (SDEs)
on $\R^{D}$ (Ito diffusions):

\begin{equation}
dx_{i}(t)=\frac{1}{N-1}\sum_{j\neq i}F(x_{i}-x_{j})dt+\sqrt{\frac{2}{\beta_{N}}}dB_{i}(t),\,\,i=1,2,...,N\label{eq:system intro}
\end{equation}
for a given parameter $\beta_{N}\in]0,\infty],$ where $B_{i}$ denotes
$N$ independent Brownian motions on $\R^{D}$ and where the sum ranges
over the $N-1$ indices where $j\neq i.$ We also allow the deterministic
case $\beta_{N}=\infty$ where the system above is an ordinary differential
equation. When $F$ is Lipchitz continuous these systems have strong
solutions, which are unique, given appropriate initial conditions
and determine the corresponding\emph{ empirical measures} 
\begin{equation}
\frac{1}{N}\sum_{i=1}^{N}\delta_{x_{i}(t)}\label{eq:empirical measure}
\end{equation}
taking values in the space $\mathcal{P}(\R^{D})$ of probability measures
on $\R^{D}.$ For general, possibly singular $F,$ such systems frequantly
arise as ``first order mean field models'' in statistical mechanics
problems in mathematical physics, biology, numerical Monte-Carlo simulations
and various other fields, where the $x_{i}$s represent the positions
of $N$ pair interacting particles (or individual/agents) on $\R^{D}$
with $F$ playing the role of the interaction force (see for example,
\cite{g=0000E4} and references therein). A classical problem is to
show that, under appropriate assumptions on $F$ and the initial data,
a\emph{ }deterministic macroscopic evolution emerges from the microscopic
dynamics in the ``many particle limit'' where $N\rightarrow\infty,$
assuming that $\beta_{N}$ has a limit: 
\[
\lim_{N\rightarrow\infty}\beta_{N}=\beta\in]0,\infty]
\]
This problem has been studied extenstively in three different settings,
ranging from the purely deterministic to the completely stochastic:
\begin{enumerate}
\item The evolution is deterministic (i.e.$\beta_{N}=\infty)$ and the initial
positions $x_{i}(t)$ are also deteministic. One then assumes that
the corresponding empirical measures of $N$ particles at $t=0$ have
a definite limit $\mu_{0},$ as $N\rightarrow\infty,$ i.e they converge
to a probability measure $\mu_{0}$ on $\R^{D},$ in a suitable topology. 
\item The evolution is deterministic, but the initial positions $x_{i}(0)$
are taken to be independent random variables with identical distribution
$\mu_{0}$
\item The noise term is present (i.e. $\beta_{N}\neq\infty)$ and the initial
positions $x_{i}(0)$ are taken to be independent random variables
with identical distribution $\mu_{0}$ 
\end{enumerate}
In the first, purely deterministic, case, the problem is to show that
there exists a curve $\mu_{t}$ in $\mathcal{P}(\R^{D})$ emanating
from $\mu_{0}$ such that, for any positive time $t,$ 
\[
\frac{1}{N}\sum_{i=1}^{N}\delta_{x_{i}(t)}\rightarrow\mu_{t}
\]
 in a given topology on $\mathcal{P}(\R^{D})$ (then the corresponding
deterministic particle system is usually said to have a \emph{mean
field limit}). In the second case the empirical measures of $N-$particle
defines, at any given time $t,$ a random measure and the problem
is then to show that the convergence above holds in law for any $t>0$
(then\emph{ propagation of chaos }is said to hold; a terminology introduced
by Kac \cite{kac}, inspired by the work of Boltzmann). For simplicity
we will simply say that the system\emph{ \ref{eq:system intro} has
a macroscopic limit $\mu_{t}$} if the convergence above holds in
all three situations (note that convergence in setting $1$ implies
convergence in the Setting $2$. 

As is well-known the setting above admits a pure PDE formulation,
not involving any stochastic calculus and it is this analytic point
of view that we will adopt here. Indeed, the laws of the SDEs \ref{eq:system intro}
define a curve $\mu_{t}^{(N)}$ of probability measures on $(\R^{D})^{N},$
which can be directly defined as the solution to a linear PDE on $(\R^{D})^{N};$
the forward Kolmgorov equation (also called the linear Fokker-Planck
equation). Accordingly, the empirical measure at time $t$ (formula
\ref{eq:empirical measure}) can be viewed as the random measure $\frac{1}{N}\sum_{i=1}^{N}\delta_{x_{i}}$
on $((\R^{D})^{N},\mu_{t}^{(N)}).$ 

For a Lipchitz continuous force term $F$ the existence of a macroscopic
limit goes back to the seminal work of McKean \cite{mc,mc2} and there
are by now various different approaches and developments (see for
example \cite{g=0000E4,sn,mmw} and for a recent review on Setting
$1$ and $2$ above see \cite{ja}). Moreover, the macroscopic limit
$\mu_{t}$ may then be characterized as the unique weak solution of
the following non-local drift-diffusion equation on $\R^{D}:$ 
\begin{equation}
\frac{d\mu_{t}}{dt}=\frac{1}{\beta}\Delta\mu_{t}-\nabla\cdot(\mu_{t}b[\mu_{t}]),\,\,\,\,b[\mu_{t}](x)=\int_{\R^{D}}F(x-y)\mu(y)\label{eq:drift-d equa intro}
\end{equation}
(often called the \emph{McKean-Vlasov equation} in the literature).
However, in many naturally occuring models the interaction force $F$
is not locally Lipchitz and even unbounded around $x=0$ and the main
purpose of the present paper is to establish the existence of a macroscopic
limit for a wide class of such singular $F,$ including strongly singular
repulsive $F$ when $D=1$ and locally bounded (but not necessarily
continuous) attractive $F,$ for any dimension $D.$

\subsection{Statement of the main results}

We will be concerned with the case when the interaction force $F$
can be realized as a gradient: 
\[
F(x)=-(\nabla w)(x)
\]
 (in the almost everwhere sense) for an even function $w.$ In particular,
this is always the case in dimension $D=1.$ Then the SDEs \ref{eq:system intro}
are often called the (overdamped)\emph{ Langevin equation. }Following
\cite{h-j1,ja} we will say that the interaction is \emph{weakly singular}
if $w(x-y)$ is continuous and\emph{ strongly singular} if its absolute
value blows-up along the diagonal. In terms of the singularity along
the diagonal our results apply, when $D=1,$ to any singular repulsive
$w$ in $L_{loc}^{1}$ (which is hence strongly singular) and when
$D\geq1$ to attractive locally Lipchitz continuous $w.$ 

We will denote by $\mathcal{P}_{2}(\R^{D})$ the space of all probability
measures on $\R^{D}$ with finite second moments, endowed with its
standard topology definied by weak convergence together with convergence
of the second moments. In other words, this is the topology on $\mathcal{P}_{2}(\R^{D})$
determined by the Wasserstein $L^{2}-$metric on $\mathcal{P}_{2}(\R^{D})$.
The interaction potential $w$ induces, in the usual way, an energy
type functional on the space $\mathcal{P}_{2}(\R^{D})$ defined by
\begin{equation}
E(\mu)=\frac{1}{2}\int_{\R^{D}\times\R^{D}}w(x-y)\mu(x)\otimes\mu(y),\label{eq:intro interaction energy}
\end{equation}
 assuming that $w\in L_{loc}^{1}.$ The corresponding free energy
functional $F_{\beta}(\mu)$ is then obtained by adding the scaled
Boltzmann entropy $H(\mu)/\beta$ to $E(\mu).$ More generally, we
will consider the case when a confining potential $V$ is included
in the system \ref{eq:system intro} (thus breaking the translational
symmetry), i.e. 
\[
F(x-y)=-(\nabla w)(x-y)-(\nabla V)(x)-(\nabla V)(y),
\]
 which amounts to replacing $w(x-y)$ with $w(x-y)+V(x)+V(y).$ We
will say that a lsc function $\psi(x)$ on $\R^{D}$ is\emph{ quasi-convex}
if it is\emph{ $\lambda-$convex,} i.e. its distributional Hessian
is bounded form below by $\lambda I,$ for some (possibly negative)
number $\lambda$ and if, in the negative case, $|\lambda|$ can be
taken arbitrarily small as $|x|\rightarrow\infty.$ For example, any
polynomial on $\R^{D}$ with leading term of the form $Cx^{2m}$ for
$C>0$ is quasi-convex in our sense, as is any perturbation of $\psi(x)$
by a compactly supported smooth function (see Section \ref{sec:Singular-pair-interactions}).
\begin{thm}
\label{thm:general pair intro}Consider the case when $D=1$ and assume
that $w(x)$ and $V(x)$ are quasi-convex on $]0,\infty[$ and $\R,$
respectively and that $w\in L_{loc}^{1}(\R).$ Then, given any limiting
initial measure $\mu_{0}$ a macroscopic evolution $\mu_{t}$ emerges.
as $N\rightarrow\infty,$ in all three settings considered above.
\end{thm}
The main novelty of the previous theorem is the establishment of propagation
of chaos for very singular interactions in the completely stochastic
Setting $3$ above, for eaxmple the strongly singular power-laws considered
in Theorem \ref{thm:power sing intro} below. Moreover, in this general
form the result also appears to be new in the deterministic settings.

Before continuing some comments on the definition of the corresponding
probability measures $\mu^{(N)}(t)$ on $\R^{N}$ are in order in
the general non-smooth setting (see Section \ref{sub:The-forward-Kolmogorov}
for more details). Concretely, but indirectly, these may be defined
by regularization: i.e. as the unique weak limit of the probability
measures $\mu_{\epsilon}^{(N)}(t)$ obtained by replacing $w$ (and
similarly $V)$ with any sequence $w_{\epsilon}$ of smooth uniformly
$\lambda-$convex functions on $[0,\infty[$ converging to $w$ as
$\epsilon\rightarrow0,$ in a suitable way (for example, increasing
to $w$). From this point of view the previous theorem can be intepreted
as saying that the limits $\epsilon\rightarrow0$ and $N\rightarrow\infty$
commute. In fact, the same proof reveals that we may as well let $\epsilon$
depend on $N$ as long as $\epsilon_{N}\rightarrow0$ as $N\rightarrow\infty.$
More directly, the curve $\mu^{(N)}(t),$ as well as the macroscopic
limit $\mu(t),$ will be intrinsically defined using the theory of
gradient flows on Wasserstein spaces, following the approach in \cite{a-g-s,a-g-z}.
In particular, the notion of generalized geodesics and generalized
$\lambda-$convexity in Wasserstein spaces, introduced in \cite{a-g-s},
plays a key role in the proof. 

An important feauture of our approach is that the convergence in Theorem
\ref{thm:general pair intro} will be established using direct variational
arguments, realizing the limit $\mu_{t}$ as the unique gradient flow
of the corresponding free energy functional $F_{\beta},$ in the sense
of evolutionary variational inequalities (EVI)\cite{a-g-z}. The corresponding
Wasserstein gradient flow of $F_{\beta}$ has previously been studied
in \cite{Carrillo-f-p,c-d-a-k-s}. This approach bypasses the delicate
issue whether the limit can be uniquely characterized as a weak solution
of the McKean-Vlasov equation \ref{eq:drift-d equa intro} in the
sense that for any $\phi\in C_{c}^{2}(\R^{D})$ 
\begin{equation}
\frac{d}{dt}\int_{\R^{D}}\mu_{t}(x)\phi(x)=\frac{1}{\beta}\int\mu_{t}(x)\Delta\phi(x)+\frac{1}{2}\int\mu_{t}(x)\otimes\mu_{t}(y)F(x-y)\cdot\left(\nabla\phi(x)-\nabla\phi(y)\right),\label{eq:weak eq intro}
\end{equation}
holds in the distributional sense when $t\in]0,\infty[$ with $\mu_{t}\rightarrow\mu_{0}$
weakly, as $t\rightarrow0$ and similarly when a potential term $V$
is added (this weak formulation is indeed equivalent to equation \ref{eq:drift-d equa intro},
when $F$ is continuous, since $F$ is odd). But once $\mu_{t}$ has
been realized as an EVI gradient flow one can invoke the subdifferential
calculus in \cite{a-g-s} (which provides a rigourous framework for
the Otto calculus \cite{ot}) to show that $\mu_{t}$ is indeed a
weak solution of the McKean-Vlasov equation \ref{eq:drift-d equa intro}.
However, it should be stressed that, in general, there is no uniqueness
result for weak solutions of such singular equations.

The following theorem illustrates our general convergence results,
in any dimension, in the case when the interaction force $F$ has
a\emph{ power-law singularity} in $\R^{D}$ in the sense that 
\[
F(x)=\frac{f(x)}{|x|{}^{\alpha}}\sigma(x)
\]
 where $f$ is a bounded $C^{2}-$smooth function and $\sigma(x):=x/|x|.$
The power-law singularity is said to be \emph{repulsive} if $f>0$
and \emph{attractive} if $f<0$ and we will refer to the case when
$f$ is constant as the \emph{model case.} 
\begin{thm}
\label{thm:power sing intro}In any dimension $D$ a macroscopic limit
emerges for the attractive power-laws with $\alpha\leq0.$ When $D=1$
the results also holds for repulsive power-laws with $\alpha\in[0,2[.$
More generally, for any $D$ the results hold when a quasi-convex
potential $V$ is included. 
\end{thm}
The repulsive case $\alpha=1$ (with $f=1$) corresponds to a logarithmic
interaction potential and has been extensively studied in random matrix
theory and free probability. The corresponding stochastic system \ref{eq:system intro}
then coincides with Dyson's Brownian motion, as recalled below. For
$\alpha>1$ the corresponding repulsive interaction potential is a
\emph{long range potential} , i.e the model pair interaction potential
is of the form $1/|x-y|^{s}$ (where $s:=\alpha-1)$ with $s<D$.
Such interaction potentials appear frequantly in the physics litterature
as pseudo/effective potentials \cite{maz}. In the model case the
corresponding McKean-Vlasov equation can be formulated as a non-local
porous medium type equation, coupled to a fractional Laplacian \cite{b-k-m,c-v}.
The long range condition is precisely what is need to make sure that
the weak McKean-Vlasov equation \ref{eq:weak eq intro} makes sense,
since it ensures that the integrand is in $L_{loc}^{1}(\R^{2D})$
(using that the the factor $\phi'(x)-\phi'(y)$ is comparable to $|x-y|$
as $x\rightarrow y).$ In this respect the previous theorem appears
to be optimal when applied to repulsive power-laws in 1D. 

The main thrust of the results in Theorem \ref{thm:power sing intro}
for attractive power-laws concern the range when $\alpha\in[-1,2[,$
i.e. whe range where the power-singularity of the interaction force
is proportional to $|x|^{\gamma}$ for $\gamma\in[0,1[$ and thus
not locally Lip continuous. The corresponding McKean-Vlasov equations
have been studied extensively in the subtle case when $\beta=\infty,$
i.e in the absense of diffusion. The equation in question is then
often called the \emph{aggregation equation }and it exhibits interesting
concentration phenomena. For example. there is, in general, no uniqueness
of weak solutions and classical solutions blow-up in a finite time
(see \cite{d-r} and references therein). The particular model case
with $D=\gamma=1$ is covered by the theory of scalar conservation
laws where the blow-up of solutions corresponds to the classical phenomen
of shock formations and where uniqueness of weak solutions only holds
for entropy solutions \cite{b-c-f}. The application of the theory
of Wasserstein gradient flows to the general aggregation was introduced
in \cite{c-d-a-k-s}, which, as recalled above, provides a canonical
solution $\mu_{t}$ of the equations for all times (generalizing the
notion of entropy solutions, as recalled in Section \ref{sub:Convex-interactions-in}). 

The results above arose as a ``spin-off effect'' of the new approach
to propagation of chaos introduced in the companion paper \cite{b-=0000F6},
motivated by complex geometry or more precisely by the construction
of Kähler-Einstien metrics on complex algebraic manifolds. In \cite{b-=0000F6}
the focus was on interaction energies $E^{(N)}(x_{1},...,x_{N})$
in $\R^{D}$ which are highly non-linear in the sense that they are
not $m-$point interactions for any finite $m.$ On the other hand
the drawback with the general result in \cite{b-=0000F6}, which is
applied to the construction of toric Kähler-Einstein metrics, is that
it requires that $E^{(N)}(x_{1},...,x_{N})$ be $\lambda-$convex
on all of $\R^{D}.$ In particular, when applied to pair interactions
of the form \ref{eq:intro interaction energy} it forces $w(x)$ to
be Lipchitz continuos around $x=0.$ The main point of the present
paper is to show how to circumvent this problem, when $D=1,$ by working
directly on the quotient space $\R^{N}/S^{N}.$ The condition that
$D=1$ then enters in the key convexity result Proposition \ref{prop:hidden convexity},
formulated in terms of optimal transportation. We also give a general
formulation of the approach in \cite{b-=0000F6} which gives a unified
approach to establish the existence of a macroscopic limit $\mu_{t}$
in all three settings above.. In fact, Theorem \ref{thm:general pair intro}
will arise as a special case of a general convergence result concerning
interaction energies $E^{(N)}(x_{1},...,x_{N})$ which, when $D=1,$
are assumed $\lambda-$convex when $x_{1}<x_{2}<...<x_{N},$ which
in particular applies to $m-$point interaction energies.

\subsection{Comparison with previous results in the singular setting}

As will be recalled below there has recently there has been remarkable
progress on the case of singular forces $F,$ but there are still
comparively few general results. One of the main problems is to single
out a class of weak solutions of an appropriate singular version of
the equation \ref{eq:drift-d equa intro} where uniqueness holds and
then to show $(i)$ a suitable compactness/continuity result when
$N\rightarrow\infty$ and $(ii)$ that any limiting curve $\mu_{t}$
obtained in the first step is contained in the class of weak solutions
where uniqueness holds (see, for example, the discussions in \cite{g=0000E4,f-h-m}).

In order to make a more precise comparison of our results to previous
results we first briefly recall the general setting of first and second
order mean field models (for a general review see \cite{ja}). The
latter models arise as Newton's $N-$body equations on the phase space
$\ensuremath{\R^{2D}}$ (with a noise term). The first order models
\ref{eq:system intro} typically arise in mathematical physics as
scaling limit of second order models containing friction/resistance,
as well as instanton (tunneling) solutions in stochastic quantization
\cite{dh}. See also \cite{cch} and references therein for applications
to mathematical biology. The first, as well as second order, models
have a rather different flavour depending on whether $F,$ viewed
as a vector field on $\R^{D}-\{0\},$ is a gradient - as in our setting
- or if it its divergence vanishes (the ``incompressive'' case).
In the stochastic setting ($\beta_{N}<\infty)$ the first order gradient
models that we consider here, i.e. overdamped Langevin processes,
are widely used as a theoretical model for Monte-Carlo Markov schemes
as used in numerics to simulate the Boltzmann-Gibbs measures associated
to a sequence of the interaction energy/Hamtiltonians $E^{(N)}$ (i.e.
the probability measures on $\R^{DN}$ proportional to $e^{-\beta_{N}E^{(N)}}dx^{\otimes N}).$
In particular, in the ``zero temperature limit'' where $\beta_{N}\rightarrow\infty$
such schemes are used to locate configurations with nearly minimal
energy $E^{(N)}$ (for example, in the model case of power laws the
interaction energy $E^{(N)}(x_{1},.,,,x_{N})$ is called the discrete
\emph{Riesz $s-$energy} and has been studied extensively in the mathematics
litterature in connection to approximation theory \cite{s-k}). In
this numerical context the problem of propagation of chaos thus amounts
to the question of whether a large-scale coherent structure should
emerge in the numerical similations, as $N\rightarrow\infty.$

\subsubsection{$D=1$}

When $D=1$ propagation of chaos in the completely stochastic setting
$3$ has been established in increasing level of generality for the
case $\alpha=1$ of the repulsive logarithmic interaction potential
$w(t)=-\log t$ with $V$ quadratic \cite{r-s,c-l}. One simplifying
feature in the logarithmic setting is that the continuiuty property
$(i)$ discussed above automatically holds since the right hand side
in the equation \ref{eq:weak eq intro} is continuous wrt the weak
topology on $\mathcal{P}(\R).$ The proofs in \cite{r-s,c-l} also
exploit further special features of the logarithmic interaction potential
and the quadratic potential $V$ leading to the uniqueness of weak
solutions of the equation \ref{eq:weak eq intro} (using a complexification
argument; see Section \ref{sub:Realization-of-}). The results in
particular apply to Dyson's Brownian motion on the space of $N\times N$
Hermitian matrices $A,$ where $x_{i}$ represent the eigenvalues
of $A$ \cite{gu} and $\beta_{N}=\sqrt{N}.$ The uniqueness of weak
solutions under the assumption that the Fourier transform of $V$
has exponential decay (in particular $V$ is real analytic) was established
in \cite[Lemma 2.6]{c-d-g} when $\beta=\infty$ and in \cite{fon}
when $\beta<\infty$ (the uniqueness results in\cite{c-d-g,fon} were
used to establish a large deviation principle when $V=0$ which in
turn implies a stronger form of propagation of chaos when $V=0$).
The propagation of chaos for a $\lambda-$convex potential $V$ was
claimed in \cite{l-l-x}, but there seems to be a gap in the proof.
Indeed, the proof in \cite{l-l-x} is based on the claim that any
weak solution $\mu_{t}$ to the corresponding equations \ref{eq:weak eq intro}
is uniquely determined by the initial data. However, the proof of
the latter claim in \cite{l-l-x} uses the formal Otto calculus, which
in order to be rigorous would require further a priori regularity
properties of $\mu_{t}.$ Instead, as explained above, the main point
of our argument is that it directly produces a EVI gradient flow $\mu_{t},$
which is a priori stronger than a weak solution to the equation \ref{eq:weak eq intro}
(see the discussion in the end of Section \ref{sub:Realization-of-}).
In the more singular case of general repulsive power-laws with $\alpha<2$
there seems to be no previous propagation of chaos results (the study
of such very singular interactions was proposed in \cite{met}). In
the setting of second order models propagation of chaos in the stochastic
setting with $D=1$ and $\alpha=0$ was settled very recently in \cite{h-s}
(see also the references in \cite{h-s} concerning the deterministic
case).

Our results also appear to be new in the purely deterministic setting
in this generality (the mean field limit in the model case when $\alpha=1$
was established in \cite{f-i-m} under stronger assumptions on the
initial data, e.g. that $\rho_{0}$ be bounded and Lipchitz continuous,
using the theory of viscosity solutions of non-local Hamilton-Jacobi
equations). On the other hand, as shown in Section \ref{sub:Comparison-with-stability},
in some situations the convergence in the deterministic setting could
also be obtained from the stability results for Wasserstein gradient
flows in \cite{a-g-s}, combined with some non-trivial Gamma-convergence
results for singular discrete interaction energies \cite{se} (but
as far as we know this has not been noticed before in the literature).

\subsubsection{$D\geq1$}

In the case of globally $\lambda-$convex interactions (as in the
case of attractive power-laws in Theorem \ref{thm:power sing intro})
and $\beta_{N}=\infty$ the convergence of the corresponding deterministic
$N-$particle system towards the Wasserstein gradient flow $\mu_{t}$
was shown in \cite{c-d-a-k-s}, using the contractivity property of
such gradient flows (see Section \ref{sub:Comparison-with-stability}
for a comparison with the present setting). This was used in \cite{c-d-a-k-s}
to show that $\mu_{t}$ aggregates into a single Dirac mass in a finite
time (starting from any compactly supported initial measure). The
main novelty of Theorem \ref{thm:power sing intro}, in the attractive
case, is thus the possibility to add noise to the attractive particle
system. For example, in the zero-temperature case $\beta_{N}\rightarrow\infty$
this yields a stochastic particle approximation scheme for constructing
the Wasserstein gradient flow solution of the corresponding aggregation
equation introduced in \cite{c-d-a-k-s}. Such stochastic particle
approximations have previously been obtained in the model case when
$D=1$ and $\alpha=0,$ using the theory of scalar conservation laws
and their entropy solutions (see \cite{b-t} for $\beta<\infty$ and
\cite{jou} when $\beta=\infty).$ The underlying determinstic microscopic
system is then the sticky particle system originating in cosmology
\cite{b-g}  and its stochastic version, with $\beta_{N}\rightarrow\infty,$
is called the adhesion model in cosmology.

Let us also briefly mention some further recent results on the general
higher dimensional setting. When $D=2$ the critical case of a power-law
$\alpha=1$ (i.e. the Newtonian case) has been studied extensively
in the divergence free case, notably in the case of the vortex model
(where $F$ is the Biot-Savart law $F(x)=\pm Jx/|x|^{2}$ with $J$
denoting rotation by $90$ degrees) motivated by the 2D Euler and
Navier-Stokes equations \cite{m-p}. For example, partial results
for the corresponding deterministic evolution (called the Helmholtz-Kirchhow
system) in the settings $1$ and $2$ were obtained in \cite{sc,sc2}
and propagation of chaos in the setting $3$ was obtained in \cite{os0,os}
using Nash type estimates (for non-negative vorticity and $\mu_{0}$
in $L^{\infty}$) and recently, in a stronger form, in \cite{f-h-m},
using the Fisher information. The results in \cite{f-h-m}, concerning
propagation of chaos in the setting $3$ (for $D=2),$ were extended
to the gradient setting in \cite{gq} under the restriction $\alpha<1$
and a partial result concerning the critical case $\alpha=1$ (appearing
in the Keller-Segel model for chemotaxis) was established in \cite{fou-j}
(saying that propagation of chaos holds for some subsequence). 

A completly different approach to mean field limits and propagation
of chaos for first order gradient models without a noise term, i.e.
in the setting $1$ and $2$ above, was introduced in \cite{cch},
using stability properties of Wasserstein $L^{p}-$distances and subtle
estimates. Under some regularity assumptions, in particular that $\mu_{t}$
is in $\mathcal{P}_{1}(\R^{D})\cap L^{p}(\R^{D})$ for $t\in[0,T],$
convergence in Setting $1$ (and similarly in Setting $2)$ was established
for any power law such that $\alpha<D/p'-1,$ where $p'$ is the dual
exponent to $p$ (in particular, this means that $\alpha<0$ when
$D=1).$ The arguments in \cite{cch} build on \cite{ha1}, where
the use of Wasserstein distances in the singular case was first introduced
and applied to the setting where the divergence of $F$ vanishes (in
the case when $p=\infty).$ See also \cite{h-j0,h-j1} for related
work on second order models where it is assumed that the condition
$\alpha<1$ holds. 

Let us finally point out that very recently the convergence of the
deterministic Helmholtz-Kirchhow system studied in \cite{sc,sc2}
was finally settled in \cite{due} (when $V=0$ and under the assumption
that the initial measure is Hölder continuous and that the initial
energies are convergent as $N\rightarrow\infty).$ The proof is based
on a modulated energy method inspired by \cite{se2} and exploits
the fact that the limiting equation $\mu_{t}$ is known to have a
Hölder continuous density (when $V=0).$ The deterministic results
in \cite{due} also apply when $D=1,$ under the condition that the
the Hölder regularity property of $\mu_{t}$ holds (which is an open
problem).

\subsection{Outline}

We start in Section \ref{sec:Preliminaries} by setting up the general
theory of Wasserstein gradient flows from \cite{a-g-s} that will
be needed in the proof of the general convergence results given in
Section \ref{sec:A-general-convergence}. As we explain the proof
of the latter result can be viewed as an analog in our setting of
the stability of Wasserstein gradient flows on a Hilbert space, established
in \cite{a-g-s,a-g-z}. General applications are discussed which which
are then developed in the case of 1D translational invariant pair
interactions in Section \ref{sec:Singular-pair-interactions} (covering
Theorem \ref{thm:general pair intro} above). In Section\ref{sub:Realization-of-}
it is shown that the corresponding Wasserstein gradient flows $\mu_{t}$
appearing in the macroscopic limit are (particular) weak solutions
of the McKean-Vlasov equation (complementing some results in \cite{Carrillo-f-p}).
Some further applications, including the case $D\geq1$ (covering
Theorem \ref{thm:power sing intro} above) are developed in Section
\ref{sec:Further-examples} and in the final Section \ref{sub:Comparison-with-stability}
a comparison with the stability result in \cite{a-g-s,a-g-z} is made,
in the deterministic setting.

\subsection{Acknowledgements}

This paper supersedes and generalizes Section $4$ in the first and
second arXiv versions of the companion paper \cite{b-=0000F6}. It
is a pleasure to thank Eric Carlen for several stimulating discussions.
Thanks also to Luigi Ambrosio for helpful comments on the first versions
of the paper and to Maxime Hauray, Dominque Lépingle, Alice Guionnet
and Jose Carrillo for providing us with references. This work was
supported by grants from the Swedish Research Council, the Knut and
Alice Wallenberg Foundation and the European Research Council.

\section{\label{sec:Preliminaries}Preliminaries}

\subsection{\label{sub:Notation}Notation}

Given a topological (Polish) space $Y$ we will denote the integration
pairing between measures $\mu$ on $Y$ (always assumed to be Borel
measures) and bounded continuous functions $f$ by 

\[
\left\langle f,\mu\right\rangle :=\int f\mu
\]
(we will avoid the use of the symbol $d\mu$ since $d$ will usually
refer to a distance function on $Y).$ In case $Y=\R^{D}$ we will
say that a measure\emph{ $\mu$ has a density,} denoted by $\rho,$
if $\mu$ is absolutely continuous wrt Lebesgue measure $dx$ and
$\mu=\rho dx.$ We will denote by $\mathcal{P}(\R^{D})$ the space
of all probability measures and by $\mathcal{P}_{ac}(\R^{D})$ the
subspace containing those with a density (which coincides with the
space $\mathcal{P}^{r}(\R^{D})$ of all regular measures as defined
in \cite{a-g-s}, in this finite dimensional situation). The \emph{Boltzmann
entropy }$H(\rho)$ taking values in $]-\infty,\infty]$) is defined
by 
\begin{equation}
H(\rho):=\int_{\R^{D}}(\log\rho)\rho dx\label{eq:def of H and I}
\end{equation}
More generally, given a reference measure $\mu_{0}$ on $Y$ the entropy
of a measure $\mu$ relative to $\mu_{0}$ is defined by 
\begin{equation}
H_{\mu_{0}}(\mu)=\int_{X^{N}}\left(\log\frac{\mu}{\mu_{0}}\right)\mu\label{eq:def of rel entropy notation}
\end{equation}
if the probability measure $\mu$ on $X$ is absolutely continuous
with respect to $\mu$ and otherwise $H(\mu):=\infty.$ Given a lower
semi-continuous (\emph{lsc}, for short) function $V$ on $Y$ and
$\beta\in]0,\infty]$ (the ``inverse temperature) we will denote
by $F_{\beta}^{V}$ the corresponding \emph{(Gibbs) free energy functional
with potential $V:$
\begin{equation}
F_{\beta}^{V}(\mu):=\int_{X}V\mu+\frac{1}{\beta}H_{\mu_{0}}(\mu),\label{eq:def of free energy of v notation}
\end{equation}
}which coincides with $\frac{1}{\beta}$ times the entropy of $\mu$
relative to $e^{-V}\mu_{0}.$

\subsection{\label{sub:Wasserstein-spaces-and}Wasserstein spaces and metrics}

We start with the following very general setup. Let $(X,d)$ be a
given metric space, which is Polish, and denote by $\mathcal{P}(X)$
the space of all probability measures on $X$ endowed with the \emph{weak
topology,} i.e. $\mu_{j}\rightarrow\mu$ weakly in $\mathcal{P}(X)$
iff $\int_{X}\mu_{j}f\rightarrow\int_{X}\mu f$ for any bounded continuous
function $f$ on $X$ (this is also called the\emph{ narrow topology}
in the probability literature). The metric $d$ on $X$ induces $l^{p}-$type
metrics on the $N-$fold product $X^{N}$ for any given $p\in[1,\infty[:$
\[
d_{p}(x_{1},...,x_{N};y_{1},...,y_{N}):=(\sum_{i=1}^{N}d(x_{i},y_{i})^{p})^{1/p}
\]
The permutation group $S^{N}$ on $N$ letters has a standard action
on $X^{N},$ defined by $(\sigma,(x_{1},...,x_{N}))\mapsto(x_{\sigma(1)},...,x_{\sigma(N)})$
and we will denote by $X^{(N)}$ and $\pi$ the corresponding quotient
and quotient projection, respectively:

\begin{equation}
X^{(N)}:=X^{N}/S^{N},\,\,\,\,\pi:\,X^{N}\rightarrow X^{(N)}\label{eq:quotient and proj}
\end{equation}
The quotient $X^{(N)}$ may be naturally identified with the space
of all configurations of $N$ points on $X.$ We will denote by $d_{(p)}$
the induced distance function on $X^{(N)},$ suitably normalized:

\[
d_{X^{(N)},l^{P}}(x_{1},...,x_{N};y_{1},...,y_{N}):=\inf_{\sigma\in S_{N}}(\frac{1}{N}\sum_{i=1}^{N}d(x_{i},y_{\sigma(i)})^{p})^{1/p}
\]
The normalization factor $1/N^{1/p}$ ensures that the standard embedding
of $X^{(N)}$ into the space $\mathcal{P}(X)$ of all probability
measures on $X:$ 
\begin{equation}
X^{(N)}\hookrightarrow\mathcal{P}(X),\,\,\,\,(x_{1},..,x_{N})\mapsto\delta_{N}:=\frac{1}{N}\sum\delta_{x_{i}}\label{eq:def of empricical measure}
\end{equation}
(where we will call $\delta_{N}$ the \emph{empirical measure}) is
isometric when $\mathcal{P}(X)$ is equipped with the\emph{ $L^{p}-$Wasserstein
metric} $d_{W^{p}}$ induced by $d$ (for simplicity we will also
write $d_{W_{p}}=d_{p}):$ 
\begin{equation}
d_{W_{p}}^{p}(\mu,\nu):=\inf_{\gamma}\int_{X\times X}d(x,y)^{p}\gamma,\label{eq:def of wasser}
\end{equation}
 where $\gamma$ ranges over all couplings between $\mu$ and $\nu,$
i.e. $\gamma$ is a probability measure on $X\times X$ whose first
and second marginals are equal to $\mu$ and $\nu,$ respectively
(see Lemma \ref{lem:isometries} below). We will denote $W^{p}(X,d)$
the corresponding\emph{ $L^{p}-$Wasserstein space}, i.e. the subspace
of $\mathcal{P}(X)$ consisting of all $\mu$ with finite $p:$th
moments: for some (and hence any) $x_{0}\in X$ 
\[
\int_{X}d(x,x_{0})^{p}\mu<\infty
\]
We will also write $W^{p}(X,d)=\mathcal{P}_{p}(X)$ when it is clear
from the context which distance $d$ on $X$ is used. 
\begin{rem}
\label{rem:transport}In the terms of the Monge-Kantorovich theory
of optimal transport \cite{v1} $d_{W_{p}}^{p}(\mu,\nu)$ is the optimal
cost to for transporting $\mu$ to $\nu$ with respect to the cost
functional $c(x,p):=d(x,y)^{p}$. Accordingly a coupling $\gamma$
as above is often called a \emph{transport plan} between $\mu$ and
$\nu$ and it said to be defined by a \emph{transport map} $T$ if
$\gamma=(I\times T)_{*}\mu$ where $T_{*}\mu=\nu.$ In particular,
if $X=\R^{n},$ $p=2$ and $\mu$ has a density, then, by Brenier's
theorem \cite{br}, the optimal transport plan $\gamma$ is always
defined by a unique transport $L_{loc}^{\infty}-$map $T(:=T_{\mu}^{\nu})$
of the form $T_{\mu}^{\nu}=\nabla(\phi(x)+|x|^{2}/2),$ where $\phi(x)+|x|^{2}/2$
is a convex function on $\R^{n}$ (optimizing the dual Kantorovich
functional) and vice versa if $\nu$ has a density. 
\end{rem}
A key point will, in the following, be played by the following isometry
properties:
\begin{lem}
\label{lem:isometries}(Three isometries)
\begin{itemize}
\item The empirical measure $\delta_{N}$ defines an isometric embedding
$(X^{(N)},d_{(p)})\rightarrow\mathcal{P}_{p}(X)$ 
\item The corresponding push-forward map $(\delta_{N})_{*}$ from $\mathcal{P}(X^{(N})$
to $\mathcal{P}(\mathcal{P}(X))$ induces an isometric embedding between
the corresponding Wasserstein spaces $W_{q}(X^{(N)},d_{(p)})$ and
$W_{q}(\mathcal{P}_{p}(X)).$ 
\item The push-forward $\pi_{*}$ of the quotient projection $\pi:X^{N}\rightarrow X^{(N)}$
induces an isometry between the subspace of symmetric measures in
$(W_{q}(X^{N},\frac{1}{N^{1/p}}d_{p})$ and the space $(W_{q}(X^{(N)},d_{(p)})$ 
\end{itemize}
\end{lem}
Let us also recall the following classical result, which is a weak
version of Sanov's theorem \cite[Theorem 6.2.10]{d-z}:
\begin{lem}
\label{lem:sanov type}Let $\mu_{0}$ be a probability measure on
$X.$ Then $(\delta_{N})_{*}\mu_{0}^{\otimes N}\rightarrow\delta_{\mu_{0}}$
in $\mathcal{P}(\mathcal{P}(X))$ weakly as $N\rightarrow\infty$ 
\end{lem}

\subsection{EVI gradient flows on the Wasserstein space}

Let $F$ be a lower semi-continuous function on a complete metric
space $(M,d).$ In this generality there are, as explained in \cite{a-g-s},
various notions of weak gradient flows $u_{t}$ for $F$ (or ``steepest
descents'') emanating from an initial point $u_{0}$ in $M.$ The
strongest form of weak gradient flows on metric spaces discussed in
\cite{a-g-s} are defined by the property that $u_{t}$ satisfies
the following \emph{Evolution Variational Inequalities (EVI) }for
some $\lambda\in\R:$

\begin{equation}
\frac{1}{2}\frac{d}{dt}d^{2}(u_{t},v)+F(u(t))+\frac{\lambda}{2}d^{2}(\mu_{t},\nu)^{2}\leq F(v)\,\,\,\,\mbox{a.e.\,\,\ensuremath{t>0,\,\,\,\forall v\in M:\,F(v)<\infty}}\label{eq:evi}
\end{equation}
together with the initial condition $\lim_{t\rightarrow0}u(t)=u_{0}$
in $(M,d).$ Then $u_{t}$ is uniquely determined by $u_{0},$ as
shown in \cite[Cor 4.3.3]{a-g-s} and we shall say that $u_{t}$ is
the \emph{EVI-gradient flow} of $F$ emanating from $u_{0}.$ As shown
in \cite{a-g-s} when $(M,d)$ is the Wasserstein space $\mathcal{P}_{2}(\R^{d})$
such gradient flows can be constructed when $F$ has certain convexity
properties that we next recall. Following \cite{a-g-s} we first recall
that a \emph{generalized geodesic} $\mu_{s}$ connecting $\mu_{0}$
and $\mu_{1}$ in $\mathcal{P}_{2}(\R^{d})$ with ``base measure''
$\nu\in\mathcal{P}_{2,ac}(\R^{d}))$ is the curve $\mu_{s}$ on $[0,1]$
with values in $\mathcal{P}_{2}(\R^{d})$ defined as the following
family of push-forward measures:
\begin{equation}
\mu_{s}=\left((1-s)T_{0}+sT_{1}\right)_{*}\nu\label{eq:general geod with density}
\end{equation}
 where $T_{i}$ is the optimal transport map (defined with respect
to the cost function $|x-y|^{2}/2)$ pushing forward $\nu$ to $\mu_{i}$
(compare Remark \ref{rem:transport}). More generally, to any triple
$(\mu_{0},\mu_{1},\nu)$ of measures in $\mathcal{P}_{2}(\R^{d})$
one can associate a corresponding notion of generalized geodesic $\mu_{s}$
in $\mathcal{P}_{2}(\R^{d})$ (which may not be uniquely determined,
unless $\nu\in\mathcal{P}_{2,ac}(\R^{d})$), using transport plans
instead of transport maps (see \cite[Def 9.2.2]{a-g-s}). Then $F$
is said to be \emph{$\lambda-$convex} \emph{along generalized geodesics}
if, given any triple $(\mu_{0},\mu_{1},\nu)$ of measures in $\{F<\infty\}$
there exists a corresponding generalized geodesic along which a certain
convexity type inequality holds (see \cite[Def 9.2.3]{a-g-s}). Anyway,
for our purposes it will be enough to consider generalized geodesics
of the form appearing in formula \ref{eq:general geod with density},
thanks to the following
\begin{prop}
\label{prop:gener conv} 
\begin{itemize}
\item The notion of $\lambda-$convexity along generalized geodesics is
preserved under $\Gamma-$convergence wrt the $\mathcal{P}_{2}(X)-$topology
\item Let $F$ be a lsc function on $\mathcal{P}_{2}(\R^{d})$ with the
property that for any $\mu\in\{F<\infty\}$ there exists a sequence
$\mu_{j}$ converging to $\mu$ in $\mathcal{P}_{2}(\R^{d})$ such
that $F(\mu_{j})\rightarrow F(\mu).$ Then $F$ is $\lambda-$convex
along any\emph{ generalized} geodesic in $\mathcal{P}_{2}(\R^{d})$
iff for any generalized geodesic $\mu_{s}$ of the form \ref{eq:general geod with density}
(i.e. with a base $\nu$ in $\mathcal{P}_{2,abs}(\R^{d})$) 
\[
F(\mu_{s})\leq(1-s)F(\mu_{0})+sF(\mu_{1})-\frac{\lambda}{2}s(1-s)d_{2}(\mu_{0},\mu_{1})^{2}
\]

\item If moreover $F$ is continuos on $\mathcal{P}(\R^{d}),$ wrt the weak
topology, then the inequality above holds for $\mu_{s}$ iff it holds
for the generalized geodesics $\mu_{s}^{(j)}$ with the same base
$\nu$ obtained by fixing two sequences $\mu_{0}^{(j)}$ and $\mu_{1}^{(j)}$
in $\mathcal{P}_{2,abs}(\R^{d})$ weakly converging to $\mu_{0}$
and $\mu_{1},$ respectively.
\end{itemize}
\end{prop}
\begin{proof}
The first and second statement is the content of \cite[Lemma 9.2.9]{a-g-s}
and \cite[Prop 9.2.10]{a-g-s}, respectively and the last statement
then follows by a standard compactness argument (just as in the proof
of \cite[Prop 9.1.3]{a-g-s}).
\end{proof}
We recall the notion of $\Gamma-$convergence used above, introduced
by De Giorgi:
\begin{defn}
\label{def:Gamma conv}A one-parameter family of $F_{h}$ of functions
on a topological space \emph{$\mathcal{P}$ is said to $\Gamma-$convergence
}to a function $F$ on $\mathcal{P}$ if 
\begin{equation}
\begin{array}{ccc}
\mu_{h}\rightarrow\mu\,\mbox{in\,}\mathcal{P} & \implies & \liminf_{h\rightarrow\infty}F_{h}(\mu_{h})\geq F(\mu)\\
\forall\mu & \exists\mu_{h}\rightarrow\mu\,\mbox{in\,}\mathcal{P}: & \lim_{h\rightarrow\infty}F_{h}(\mu_{h})=F(\mu)
\end{array}\label{eq:def of gamma conv}
\end{equation}
In the present setting $\mathcal{P}=\mathcal{P}_{2}(Y,d)$ and we
will then say that $F_{h}$\emph{ $\Gamma-$converges to $F$ strongly}
if, in the lower bound above holds also for every $\mu_{h}$ converging
in the weak topology together with a uniform bound on the second moments
(the terminology strong is non-standard)

We will need the following equivariant generalization of \cite[Thm 11.2.1]{a-g-s}
which constructs an EVI gradient flow as a \emph{Minimizing Movement,}
i.e. as a limit of a time discretized Minimizing Movement scheme:\end{defn}
\begin{thm}
\label{thm:existence of evi}Let $G$ be a compact group acting by
isometries on $\R^{d}$ and $F$ a $G-$invariant lsc real-valued
functional on $\mathcal{P}_{2}(\R^{d})$ which is $\lambda-$convex
along $G-$invariant generalized geodesics and satisfies the following
coercivity property: there exist constants $\tau_{*},C>0$ and $\mu_{*}\in\mathcal{P}_{2}(\R^{d})$
such that 
\begin{equation}
F(\cdot)\geq-\frac{1}{\tau_{*}}d_{2}(\cdot,\mu_{*})^{2}-C\label{eq:coercivity type condition}
\end{equation}
Then there is a unique solution $\mu_{t}$ to the EVI-gradient flow
of $F,$ emanating from any given $G-$invariant $\mu_{0}$ in $\overline{\left\{ F<\infty\right\} }$
and $\mu_{t}$ remains $G-$invariant for any $t>0.$ Moreover, the
flow is $\lambda-$contractive: 
\[
d_{2}(\mu_{t},\nu_{t})\leq e^{-\lambda t}d_{2}(\mu_{0},\nu_{0})
\]
and 
\[
d_{2}(\mu_{t},\mu_{t}^{(\tau)})\leq C|\tau|^{1/2}F(\mu_{0})\,\,\,\,t\in[0,T]
\]
 where $\mu_{t}^{(\tau)}$ is the corresponding minimizing movement
with timestep $\tau$ and $C$ is a constant only depending on $\lambda,$
$T$ \end{thm}
\begin{proof}
This is a straightforward generalization of the proof of \cite[Thm 11.2.1]{a-g-s}.
To see this we set $M:=\mathcal{P}_{2}(\R^{d})^{G}$ viewed as a closed
subspace of the metric space $\mathcal{P}_{2}(\R^{d})^{G}.$ Restricting
$d_{2}$ to $M$ gives a complete metric space $(M,d_{2}).$ By Theorem
4.0.4 in \cite{a-g-s} we just have to verify the following two conditions
on $(M,d_{2}):$ for any choice of $\mu_{0},\mu_{1}$ and $\nu$ in
$M$ there exists a curve $\gamma_{t}$ in $\mathcal{P}_{2}(\R^{d})^{G}$
connecting $\mu_{0}$ and $\mu_{1}$ such that 
\begin{itemize}
\item $d_{2}^{2}(\gamma_{t},\mu)$ is $\lambda-$ convex wrt $t$ for some
$\lambda>0$ 
\item $F(\gamma_{t})$ is $\lambda-$convex wrt $t$
\end{itemize}
As shown in\cite{a-g-s} such a curve $\gamma_{t}$ exists in the
space $\mathcal{P}_{2}(\R^{d})$ and may be taken as a generalized
geodesic connecting $\mu_{0}$ and $\mu_{1}$ with base $\nu.$ Accordingly,
all we have to do is to verify that if $\mu_{0},\mu_{1}$ and $\nu$
are $G-$invariant, then the corresponding generalized geodesic may
be taken to be $G-$invariant for all $t\in[0,1].$ But this follows
from the dual Kantorovich formulation of the optimal transport problem
\cite{v1}. To see this final point, note that by compactness $G$
can be embedded into $O(d,\R)$, with the natural action of $O(d,\R)$
around some fixed point. But for $A\in O(d,\R)$, we have that $(\phi\circ A)^{*}=\phi^{*}\circ A$,
where the star denotes Legendre transformation. \end{proof}
\begin{rem}
We recall that it follows from the results in \cite{a-g-s} that the
EVI-gradient flow above has a number of further properties. For example,
the flow defines a semi-group and is $1/2-$Hölder continuous as a
map from any fixed time interval $[0,T]$ into $\mathcal{P}_{2}(\R^{d})$
and Lipschitz continuous on any fixed bounded open time-interval.
We also note that the corresponding curve $[\mu_{t}]\in W_{2}(\R^{d}/G)$
is the unique EVI-gradient flow of $F$ viewed as a function on $W_{2}(\R^{d}/G),$
as follows immediately from the the natural isometry between $W_{2}(\R^{d})^{G}$
and $W_{2}(\R^{d}/G).$ 
\end{rem}
We briefly recall the construction of De Giorgi's \emph{minimizing
movement} scheme in a general metric space $(M,d),$ which can be
seen as a variational formulation of the (back-ward) Euler scheme
\cite[Chapter 2]{a-g-s}).. Consider the fixed time interval $[0,T]$
and fix a (small) positive number $\tau$ (the ``time step''). In
order to define the ``discrete flow'' $u_{j}^{\tau}$ corresponding
to the sequence of discrete times $t_{j}:=j\tau,$ where $t_{j}\leq T$
with initial data $u_{0}$ one proceeds by iteration: given $u_{j}\in M$
the next step $u_{j+1}$ is obtained by minimizing the following functional
on $(M,d):=W_{2}(\R^{d}):$ 
\[
u\mapsto J_{j+1}(u):=\frac{1}{2\tau}d(u,u_{j})^{2}+F(u)
\]
Finally, one defines $u^{\tau}(t)$ for any $t\in[0,T]$ by setting
$u^{\tau}(t_{j})=u_{j}^{\tau}$ and demanding that $u^{\tau}(t)$
be constant on $]t_{j},t_{j+1}[$ and right continuous (we are using
a slightly different notation than the one in \cite[Chapter 2]{a-g-s}).

The following result goes back to McCann \cite{mcCa} (see also \cite{a-g-s}
for various elaborations):
\begin{lem}
\label{lem:generalized conv} The following functionals are lsc and
$\lambda-$convex along any generalized geodesics in $\mathcal{P}_{2}(\R^{d})$:.
\begin{itemize}
\item The ``potential energy'' functional $\mathcal{V}(\mu):=\int V\mu,$
defined by a given lsc $\lambda-$convex and lsc function $V$ on
$\R^{d}$ (and the converse also holds)
\item The functional $\mu\mapsto\int V_{N}\mu^{\otimes N}$ defined by a
given $\lambda-$convex function $V_{N}$ on $\R^{dN}$ and in particular
the ``interaction energy'' functional 
\[
\mathcal{W}(\mu):=\int W(x-y)\mu(x)\otimes\mu(x)
\]
 defined by a given lsc $\lambda-$convex function $W$ on $\R^{d}.$
\item The Boltzmann entropy $H(\mu)$ (relative to $dx$) is lsc and convex
along any generalized geodesics.
\end{itemize}
In particular, for any $\lambda-$convex function $V$ on $\R^{d}$
the corresponding free energy functional $F_{\beta}^{V}$ (formula
\ref{eq:def of free energy of v notation}) is $\lambda-$convex along
generalized geodesics, if $\beta\in]0,\infty].$ \end{lem}
\begin{rem}
\label{rem:Otto}According to the Otto calculus \cite{ot} the EVI
gradient flow on $\mathcal{P}_{2}(\R^{d})$ of a sufficently regular
functional $F$ satisfies, if $\mu_{t}$ has has a smooth positive
density $\rho_{t},$ the evolution equation 
\begin{equation}
\frac{\partial\rho_{t}(x)}{\partial t}=\nabla_{x}\cdot(\rho v_{t}(x)),\,\,\,\,\,\,v_{t}(x)=\nabla_{x}\frac{\partial F(\rho)}{\partial\rho}_{|\rho=\rho_{t}}\label{eq:formal gradient flow}
\end{equation}
As shown in \cite{a-g-s} these equations still hold in the weak sense
of distributions under apprioriate assumptions on $E$, for the EVI
gradient flow solution $\mu_{t}$ (see Section \ref{sub:Realization-of-}).
In particular, when $F$ is of the free energy form $F=E(\mu)+H(\mu)/\beta$
\begin{equation}
\frac{\partial\rho_{t}(x)}{\partial t}=\frac{1}{\beta}\Delta_{x}\rho_{t}(x)+\nabla_{x}\cdot(\rho v_{t}(x)),\,\,\,\,\,\,v_{t}(x)=\nabla_{x}\frac{\partial E(\rho)}{\partial\rho}_{|\rho=\rho_{t}}\label{eq:formal gradient flow for free}
\end{equation}
coincides with the McKean-Vlasov equation \ref{eq:drift-d equa intro}
when $E=\mathcal{W}$ for $W$ smooth. Moreover, when $E=\mathcal{V}$
the corresponding evolution equation is the linear Fokker-Planck equation
associated to the potential $V,$ as first shown when $V$ is smooth
in the seminal work \cite{j-k-o}.
\end{rem}

\subsection{\label{sub:The-forward-Kolmogorov}The defining laws on $X^{N}$
and the mean (free) energy $\mathcal{E}^{(N)}$ (and $\mathcal{F}_{N})$}

Let $X$ be the Euclidean space $\R^{d}$ and $V$ a smooth (and say
coercive) function on $\R^{D}.$ The SDE 
\begin{equation}
dx=-\nabla Vdt+dB/\beta\label{eq:sde in section prel}
\end{equation}
where $x(0)$ is a vector of iid variables with law $\mu_{0}$ (as
before we allow the ODE case $\beta=\infty)$ defines, for any fixed
$T,$ a probability measure $\eta_{T}$ on the space of all maps $[0,T]\rightarrow X$
(see for example \cite{sn} and reference therein). For $t$ fixed
we can thus view $x(t)$ as an $X-$valued random variable on the
latter probability space. Its law gives a curve of probability measures
on $X$ of the form $\mu_{t}=\rho_{t}dx,$ where the density $\rho_{t}$
satisfies the corresponding \emph{forward Kolmogorov equation}: 
\begin{equation}
\frac{\partial\rho_{t}}{\partial t}=\frac{1}{\beta}\Delta\rho_{t}+\nabla\cdot(\rho_{t}\nabla V),\label{eq:forward kolm eq}
\end{equation}
(also called the linear\emph{ Fokker-Planck equation}). Anyway, for
our purposes we may as well forget about the SDE \ref{eq:sde in section prel}
and take the forward Kolmogorov equation \ref{eq:forward kolm eq}
on $X$ as our the starting point. As recalled above (Remark \ref{rem:Otto})
the latter evolution equation can be interpreted as the gradient-flow
on the Wasserstein space $W_{2}((\R^{D})^{N}),$ of the corresponding
free energy functional. 

In our setting we will take $V:=E^{(N)}(x_{1},...,x_{N})$ for a given
symmetric function on $X:=(\R^{D})^{N}.$ Following standard terminology
in statistical mechanics we will call the corresponding (scaled) linear
functional $\mathcal{E}^{(N)}$ on $\mathcal{P}((\R^{d})^{N}),$ defined
by 
\begin{equation}
\mathcal{E}^{(N)}(\mu_{N}):=\frac{1}{N}\int_{(\R^{D})^{N}}E^{(N)}\mu_{N}\label{eq:def of mean energy}
\end{equation}
 the \emph{mean energy.} Similarly, the corresponding mean free energy
$\mathcal{F}^{(N)},$ at inverse temperature $\beta_{N},$ is defined
by
\begin{equation}
\mathcal{F}^{(N)}(\mu_{N}):=\frac{1}{N}\int_{(\R^{D})^{N}}E^{(N)}\mu_{N}+\frac{1}{\beta_{N}N}H(\mu_{N}):=\mathcal{E}^{(N)}(\mu_{N})+\frac{1}{\beta_{N}}\mathcal{H}^{(N)}(\mu_{N})\label{eq:def of mean free energy}
\end{equation}
where the scaled Boltzmann entropy $\mathcal{H}^{(N)}(\mu_{N})$ on
$\mathcal{P}((\R^{D})^{N})$ is called the\emph{ mean entropy.}

More generally, we will allow $E^{(N)}$ to be singular, but we will
make assumptions (such as $\lambda-$convexity modulo $S^{N})$ ensuring,
by Theorem \ref{thm:existence of evi}, that the Wasserstein gradient
flow of $F_{N}$ is well-defined, giving a curve of probability measures
$\mu^{(N)}(t)$ on $(\R^{D})^{N}.$ In the deterministic setting $(\beta=\infty)$
it can be shown (see Section \ref{sub:Deterministic-limits-for})
that this approach is consistent with the classical notion of a strong
solution, as long as such a solution exists (for example, when particles
do not collide for $t>0)$. In the stochastic setting $(\beta<\infty)$
we will not attempt to give any definition of a strong solution to
the stochastic equation \ref{eq:sde in section prel}, since such
solutions do not play any role in our proofs. On the other hand the
present definition of $\mu^{(N)}(t)$ is stable under monotone regularizations
of $E^{(N)}$ (preserving $\lambda)$ and thus coincides with any
other definition with similar stablily properties. We refer to \cite{a-g-z,ce,c-l}
for the notion of strong solutions under convexity assumptions as
above and \cite{st-v} for general results about strong solutions
of SDEs with locally bounded drifts (see also Remark \ref{rem:monotone maps}).

\section{\label{sec:A-general-convergence}A general convergence result}

In section we will give a more general formulation of the propagation
of chaos result in \cite[Theorem 1.1]{b-=0000F6}, by exploiting the
$S^{N}-$symmetry (which will be cruical in the applications to strongly
singular pair potentials). The present formulation also has the virtue
of also applying in the purely deterministic setting. The proof could
be given essentially by repeating the argument in \cite{b-=0000F6}.
But here we give a slightly different proof which can be seen as an
analoge in our setting of the stability result of gradient flows on
the Wasserstein space of a Hilbert space $Y$ in \cite{a-g-s,a-g-z}
(recalled in Theorem \ref{thm:-stab of gradient flo} below). The
new difficulties that arise in our setting is that 
\begin{itemize}
\item The Gamma convergence of the corresponding functionals only holds
in a restricted ``relative'' sense
\item The space $Y$ is $\mathcal{P}_{2}(\R^{D}),$ which is not a Hilbert
space (in particular, there are no general existence results for EVI
gradient flows on $\mathcal{P}_{2}(Y),$ nor error estimates and convergence
results for general minimizing movements).
\end{itemize}
These diffuclties will be handled by exploiting the fact that the
the limiting functional $\mathcal{F}$ is linear wrt the ordinary
affine structure on $\mathcal{P}_{2}(\R^{D})$ and using the isometry
properties of the embeddings in Lemma \ref{lem:isometries}.

\subsection{\label{sub:The-assumptions-on}The assumptions on $E^{(N)}$}

Set $X=\R^{n}$ and denote by $d$ the Euclidean distance function
on $X.$ In the following $E^{(N)}$ will denote a lsc symmetric,
i.e. $S_{N}-$invariant, sequence of functions in $L_{loc}^{1}(X^{N})$
and we will make the following assumptions, where $\mathcal{E}^{(N)}$
denotes the corresponding mean energies (formula \ref{eq:def of mean energy}): 
\begin{enumerate}
\item (``Convergence of the mean energies''): There exists a lsc functional
$E(\mu)$ on $\mathcal{P}_{2}(X)$ with the property that $\{E<\infty\}$
is dense in $\mathcal{P}_{2}(X)$, $\{E<\infty\}\cap\{H<\infty\}\neq\emptyset$,
and such that for any sequence of symmetric probability measures $\mu^{(N)}$
on $X^{N}$ satisfying $\Gamma_{N}:=(\delta_{N})_{*}\mu^{(N)}\rightarrow\Gamma$
weakly in $\mathcal{P}(\mathcal{P}(X))$ and with a uniform bound
on the second moments we have 
\[
\int_{\mathcal{P}(X)}E(\mu)\Gamma(\mu)\leq\liminf_{N\rightarrow\infty}\mathcal{E}^{(N)}(\mu^{(N)})
\]
and for any $\mu\in\mathcal{P}_{2}(X)$ we have that 
\[
\limsup_{N\rightarrow\infty}\mathcal{E}^{(N)}(\mu^{\otimes N})\leq E(\mu)
\]

\item (``Convexity of the mean energies''): The mean energy functional
$\mathcal{E}^{(N)}$ on $\mathcal{P}(X^{N})^{S^{N}}$is $\lambda-$convex
along generalized geodesics in $\mathcal{P}(X^{N})^{S^{N}}.$
\item (``Coercivity'') There exists a constant $C$ such that 
\[
E^{(N)}(x_{1},...x_{N})\geq-C\frac{|x_{1}|^{2}+...+|x_{1}|^{2}}{N}-C
\]
(or equivalently, that $\mathcal{E}^{(N)}\geq-C((\delta_{N})_{*})^{*}d^{2}(\cdot,\Gamma_{0})-C$
for a fixed element $\Gamma_{0}$ in $\mathcal{P}_{2}(\mathcal{P}_{2}(X)))$ \end{enumerate}
\begin{lem}
The functional $E(\mu)$ is \textup{$\lambda-$convex along generalized
geodesics in $\mathcal{P}_{2}(X)$ and coercive.}\end{lem}
\begin{proof}
We first observe that taking $\mu^{(N)}=\mu^{\otimes N}$ and using
Sanov's theorem (or Lemma \ref{lem:sanov type}) gives 
\[
\lim\mathcal{E}^{(N)}(\mu^{\otimes N})=\lim_{N\rightarrow\infty}\frac{1}{N}\int_{X^{N}}E^{(N)}\mu^{\otimes N}=E(\mu)
\]
Fix $\nu,\mu_{0},\mu_{1}\in\mathcal{P}_{2}(X)$, and let $\mu_{t}$
be the generalized geodesic with base measure $\nu$. But then $\mu_{t}^{\otimes N}$
is the (symmetric) generalized geodesic with base $\nu^{\otimes N}$
connecting $\mu_{0}^{\otimes N},\mu_{1}^{\otimes N}$, and the convexity
statement follows from the convexity assumption on $\mathcal{E}^{(N)}$.
The coercivity of $E$ follows from the coercivity assumption on $\mathcal{E}^{(N)}$
by letting $\Gamma_{0}=\delta_{\mu_{*}}$ and writing

\[
\mathcal{E}^{(N)}(\mu^{\text{\ensuremath{\otimes}N}})\geq-Cd_{2}(\delta_{N*}\mu^{\otimes N},\delta_{\mu_{*}})^{2}-C\rightarrow-Cd_{2}(\mu,\mu_{*})\text{\texttwosuperior}-C,
\]
 By Sanov's theorem and the isometry properties of Lemma \ref{lem:isometries}.
\end{proof}

\subsection{Formulation of the general convergence results}
\begin{thm}
\label{thm:general}Let $E^{(N)}$ be a sequence of functions on $(\R^{d})^{N}$
satisfying the assumptions above and let $\mu^{(N)}$ be a sequence
of symmetric probability measures on $(\R^{d})^{N}$ such that 
\[
\Gamma_{N}:=(\delta_{N})_{*}\mu^{(N)}\rightarrow\delta_{\mu_{0}}
\]
in $W_{2}(\mathcal{P}_{2}(\R^{d}))$ as $N\rightarrow\infty,$ where
$\mu_{0}\in\mathcal{P}_{2}(\R^{d}).$ Further assume that $\mu^{(N)}\in\{F_{\beta}^{(N)}<\infty\}$,
$\mu_{0}\in\overline{{F_{\beta}<\infty}}$. Then the EVI gradient
flow solution $\mu^{(N)}(t)$ of the corresponding forward Kolmogorov
equation \ref{eq:forward kolm eq} at inverse temperature $\beta_{N}$
on $(\R^{d})^{N}$ with initial data $\mu^{(N)}(0)=\mu^{(N)}$ satisfies
\[
\Gamma_{N}(t):=(\delta_{N})_{*}\mu^{(N)}(t)\rightarrow\delta_{\mu_{t}},
\]
in $W_{2}(\mathcal{P}_{2}(\R^{d}))$ as $N\rightarrow\infty,$ where
$\mu_{t}$ is the EVI gradient flow on $\mathcal{P}_{2}(\R^{d})$
of the corresponding free energy type functional $F_{\beta}(\mu).$
\end{thm}

\subsection{The proof of Theorem \ref{thm:general}}

In the proof we will need a modifed form of strong Gamma-convergence
that we call\emph{ Gamma convergence relative to the subset }$\mathcal{D}\subset\mathcal{P}$
defined by only requiring that the equality in the second condition
in Definition \ref{def:Gamma conv} holds for $\mu\in\mathcal{D}$
(and similarly for strong Gamma-convergence). We will also say that
a functional $\mathcal{F}$ on $\mathcal{P}_{2}(Y,d)$ admits \emph{mininimizing
movements relative to $\mathcal{D}\subset\mathcal{P}_{2}(Y,d)$} if
for any given $\Gamma_{0}\in\mathcal{D}$ there exists a continuous
curve $\Gamma(t)\in\mathcal{D}$ emanating from $\Gamma_{0},$ which
can be realized as the limit, as $\tau\rightarrow0,$ of time discrete
minimizing movements $\Gamma_{\tau}(t)\in\mathcal{D},$ where $\Gamma_{\tau}(t)$
is assumed to be uniquely determined by $\Gamma_{0}.$

Now set $Y:=\mathcal{P}_{2}(\R^{D}).$ By embedding $\mathcal{P}_{2}((\R^{D})^{N}/S_{N})$
isometrically into $\mathcal{P}_{2}(Y),$ using the push-ward map
$(\delta_{N})_{*}$ we can and identify the mean free energies $\mathcal{F}^{(N)}$
with functionals on $\mathcal{P}_{2}(Y),$ extended by $\infty$ to
all of $\mathcal{P}_{2}(Y).$ We will take $\mathcal{D}$ to be subset
of all Dirac measures $\Gamma$ in $\mathcal{P}_{2}(Y),$ i.e. $\Gamma=\delta_{\mu}$
for some $\mu\in Y.$
\begin{lem}
\label{lem:gamma relative}If Assumption $1$ in Section \ref{sec:A-general-convergence}
holds, then the mean free energies $\mathcal{F}^{(N)}$ strongly Gamma-converges
to the lsc linear functional $\mathcal{F}(\Gamma)$ on $\mathcal{P}_{2}(Y),$
relative to $\mathcal{D}.$\end{lem}
\begin{proof}
The lower bound follows directly from Assumption $1$ together with
the fact that the mean entropy functionals satisfy the lower bound
in the Gamma convergence (by subadditivity \cite{r-r}; see also Theorem
5.5 in \cite{h-m} for generalizations). To prove the upper bound,
we first consider the case $\beta<\infty$ and fix an element $\Gamma$
of the form $\delta_{\mu}$. We may then take the approximating sequence
to be for the form $(\delta_{N})_{*}\mu^{\otimes N}.$ Then the required
convergence follows from Assumption $1$ together with the basic property
$\mathcal{H}^{(N)}(\mu^{\otimes N})=H(\mu).$ When $\beta=\infty$,
since it may happen that $H(\mu)=\infty$, we must first regularize
the measure $\mu$. Fix an arbitrary measure $\mu_{0}$ such that
$\mathcal{F}(\mu)<\infty$, $H(\mu)<\infty$, let $\mu_{1}=\mu$,
and for $t\in[0,1]$ let $\mu_{t}$ be the displacement interpolation.
Then $H(\mu_{t})<\infty$ for all $t\in[0,1[$ (see Lemma \ref{lem:stability-wrt-beta}),
and by the above argument it then holds that $\mathcal{F}^{(N)}(\mu_{t}^{\otimes N})\rightarrow\mathcal{F}(\mu_{t})$.
But by $\lambda$-convexity and lower semicontinuity, $t\mapsto\mathcal{F}(\mu_{t})$
is continuous, and thus $\mathcal{F}(\mu_{t})\rightarrow\mathcal{F}(\mu)$
as $t\rightarrow1$. By a diagonal argument we can then find a sequence
$t_{N}$ such that $\mathcal{F}^{(N)}(\mu_{t_{N}}^{\otimes N})\rightarrow\mathcal{F}(\mu)$,
completing the proof. \end{proof}
\begin{lem}
\label{lem:min mov relative to}Let $F$ be a lsc functional on a
metric space $(Y,d)$ with the property that at for any given element$y_{0}$
in $(Y,d)$ there is an EVI gradient flow (with parameter $\lambda)$
emanating from $y_{0}$ and which can be realized as a limit of discrete
minimizing movements. Denote by $\mathcal{F}$ the lsc linear functional
on $\mathcal{P}_{2}(Y,d)$ associated to $F.$ Then, for any given
element $\Gamma_{0}$ of the form $\Gamma_{0}=\delta_{y_{0}}$ in
$\mathcal{P}_{2}(Y,d)$ there is an EVI gradient flow (with parameter
$\lambda)$ emanating from $\Gamma_{0},$ namely $\Gamma_{0}:=\delta_{y_{t}},$
where $y_{t}$ is the EVI gradient flow of $F$ emanating from $y_{0}.$
In particular, $\mathcal{F}$ admits mininimizing movements relative
to $\mathcal{D}.$\end{lem}
\begin{proof}
This is an abstract version of the third point in Lemma \ref{lem:strong sing}
below and is proved in exactly the same way using that $\mathcal{F}$
is linear. 
\end{proof}
Now we discretize time, with mesh $\tau,$ and consider the minimizing
movements $\Gamma_{N}^{\tau}(t)$ and $\Gamma^{\tau}(t)$ of $\mathcal{F}^{(N)}$
and $\mathcal{F},$ respecively. By Lemma \ref{lem:gamma relative}
we have that, for any fixed $\Gamma\in\mathcal{P}_{2}(Y)$ 
\[
\mathcal{J}^{(N)}(\cdot):=\frac{1}{2\tau}d(\cdot,\Gamma)^{2}+\mathcal{F}^{(N)}\rightarrow\mathcal{J}(\cdot):=\frac{1}{2\tau}d(\cdot,\Gamma)^{2}+\mathcal{F}
\]
 in the sense of relative strong Gamma-convergence, as $N\rightarrow\infty.$
But then it follows from basic properties of Gamma-convergence, using
that $\mathcal{J}^{(N)}(\cdot)$ is uniformly coercive and compactness
properties in the Wasserstein space,  that the corresponding minimizers
converge. Hence, starting with $\Gamma_{0}:=\delta_{\mu_{0}}$ it
follows by induction that 
\[
\lim_{N\rightarrow\infty}\Gamma_{N}^{\tau}(t)=\Gamma^{\tau}(t),
\]
 where $\Gamma^{\tau}(t)\in\mathcal{D}\subset\mathcal{P}_{2}(Y)$
at any fixed time $t.$ 

To conclude the proof of Theorem \ref{thm:general} we just have to
make sure that the error terms appearing when comparing $\Gamma_{N}^{\tau}(t)$
with $\Gamma_{N}(t),$ when $\tau\rightarrow0$ can be uniformly controlled
as $N\rightarrow\infty.$

\emph{Step 1:} First assume that $\mu_{0}\in\{F_{\beta}<\infty\}$
and $\limsup F_{\beta}^{(N)}(\mu^{(N)})<\infty$. By Theorem \ref{thm:existence of evi}
the gradient flow $\mu_{t}$ of $F$ emanating from a given $\mu_{0}$
exists and is uniquely determined. We let $\Gamma_{t}:=\delta_{\mu_{t}}$
be the corresponding flow on $\mathcal{P}_{2}(\mathcal{P}_{2}(X)).$
Consider the fixed time interval $[0,T]$ and fix a small time step
$\tau>0.$ Denote by $\mu^{\tau}(t)$ the discretized minimizing movement
of $F(\mu)$ with time step $\tau$ and set $\Gamma_{t}^{\tau}:=\delta_{\mu_{t}^{\tau}}.$
For any fixed $t\in]0,T[$ we then have, by the triangle inequality, 

\[
d(\Gamma_{N}(t),\Gamma(t))\leq d(\Gamma_{N}(t),\Gamma_{N}^{\tau}(t))+d(\Gamma(t),\Gamma^{\tau}(t))+d(\Gamma_{N}^{\tau}(t),\Gamma^{\tau}(t))
\]
By the isometry property in Lemma \ref{lem:isometries} and the assumed
convexity properties we have, by \cite[Theorem 4.0.4,4.0.7, p.79]{a-g-s},
that $d(\Gamma_{N}(t),\Gamma_{N}^{\tau}(t))\leq C\tau^{1/2}$ (uniformly
in $N)$ and $d(\Gamma(t),\Gamma^{\tau}(t))\leq C\tau^{1/2}.$ Moreover,
by \ref{lem:min mov relative to},$\lim_{N\rightarrow\infty}d(\Gamma_{N}^{\tau}(t),\Gamma^{\tau}(t))=0$
for any fixed $\tau.$ Hence, letting first $N\rightarrow\infty$
and then $\tau\rightarrow0$ gives $\lim_{N\rightarrow\infty}d(\Gamma_{N}(t),\Gamma(t))=0,$
which concludes the proof.

\emph{Step 2:} The case when $F(\mu_{0})<\infty$

Set $\nu_{0}^{(N)}:=\mu_{0}^{\otimes N}$ and denote by $\nu_{t}^{(N)}$the
EVI gradient flow of $\mathcal{F}^{(N)}$ emanating from $\nu_{0}^{(N)}.$
Let $\epsilon_{N}(t):=N^{-1}d_{2}(\nu_{t}^{(N)},\mu_{t}^{(N)}).$
By Lemma \ref{lem:sanov type} and the isometry properties in Lemma
\ref{lem:isometries} $\epsilon_{N}(0)\rightarrow0.$ Hence, by the
$\lambda-$contractivity in Theorem \ref{thm:existence of evi} $\epsilon_{N}(t)\leq e^{-\lambda t}\epsilon_{N}(0)\rightarrow0$
for any fixed postive $t>0.$ But then the desired convergence follows
from the previous step, using the triangle inequality. 

\emph{Step 3: }The case of a general $\mu_{0}\in\mathcal{P}_{2}(\R^{D})$

By assumption there exists a sequence $\mu_{j,0}$ in $\mathcal{P}_{2}(\R^{D})$
such that $d(\mu_{j,0},\mu_{0})\leq1/j$ and $F(\mu_{j,0})<\infty.$
We define $\nu_{j}^{(N)}(t)$ as before, up to replacing $\mu_{0}(:=\nu_{\infty})$
with $\nu_{j}$ and then set $\Gamma_{j,N}(t):=(\delta_{N})_{*}\nu_{j}^{(N)}(t)$.
By the triangle inequality we have, for any fixed $j,$ 
\[
d(\Gamma_{N}(t),\Gamma(t))\leq d(\Gamma_{j,N}(t),\Gamma_{j}(t))+d(\Gamma(t),\Gamma_{j}(t))+d(\Gamma_{N}(t),\Gamma_{j,N}(t)),
\]
 where the first term tends to zero as $N\rightarrow\infty$ by the
previous step. By construction the second term satisfies $d(\Gamma(0),\Gamma_{j}(0))\leq1/j$
and hence, by $\lambda-$contractivity, $d(\Gamma(t),\Gamma_{j}(t))\leq e^{-\lambda t}/j.$
Similarly, using the triangle inequalty again, the third term satisfies

\[
d(\Gamma_{N}(0),\Gamma_{j,N,}(0))\leq N^{-1}d_{2}(\mu^{(N)}(0),\mu_{0}^{\otimes N})+1/j:=\epsilon_{N}+1/j
\]
where$\epsilon_{N}\rightarrow0$ as $N\rightarrow\infty$ (as in Step
1). Hence, by $\lambda-$contractivity, $d(\Gamma_{N}(t),\Gamma_{j,N,}(t))\leq e^{-\lambda t}(\epsilon_{N}+1/j).$
Accordingly, letting first $N\rightarrow\infty$ and then $j\rightarrow\infty$
concludes the proof.

\subsection{\label{sub:Applications}General structure of the applications}

For applications to the purely deterministic Setting $1$ (described
in the introduction of the paper) one simply take $\mu^{(N)}$ to
be the normalized $S_{N}-$orbit in $X^{N}$of the Dirac measure supported
at $(x_{1}(0),....,x_{N}(0)).$ In the Settting $2$ and $3$ in the
introduction of the paper one takes $\mu^{(N)}=\mu^{\otimes N}.$
Before developing these applications in some particular settings in
Sections \ref{sec:Singular-pair-interactions}, \ref{sec:Further-examples},
we make some remarks about more general situations where the assumptions
in Section \ref{sub:The-assumptions-on} hold. First of all, as shown
in \cite{b-=0000F6}, the assumptions hold when 
\begin{equation}
\frac{1}{N}E^{(N)}(x_{1},x_{2},....,x_{N})=E(\delta_{N})+o(1),\label{eq:limit of EN intro}
\end{equation}
when $E$ is uniformly Lipchitz continuous and $\lambda-$convex and
the error term tends to zero, as $N\rightarrow\infty$ (for $x_{i}$
uniformly bounded). But the main point in the present paper is that
the assumptions are also satisified in some naturally occuring very
singular situations. For example, Assumption $1$ is satisfied if
one starts with a ``polynomial'' functional $E(\mu)$ on $\mathcal{P}(X),$
i.e. 
\[
E(\mu)=\sum_{m=1}^{M}\int_{X^{m}}w_{m}\mu^{\otimes m}
\]
where $w_{m}$ are assumed upper semi-continuous functions $X^{m}\rightarrow]-\infty,\infty]$
in $L_{loc}^{1}(X^{m}),$ which are smooth (or continuous) on the
open subset of configurations in $X^{m}$ where no two points coincide.
Then one can then definie a ``renormalized'' interaction $N-$particle
interaction $E^{(N)}(x_{1},...x_{N})$ by setting 
\[
E^{(N)}(x_{1},...x_{N}):=\frac{1}{N^{(m-1)}}\sum_{m=1}^{M}\sum_{I}w_{m}(x_{i_{1}},...,x_{i_{m}}),
\]
 where the inner sum rums over all multiindices $I=(i_{1_{i}},...,i_{m})$
of length $m$ and with the property that no two indices of $I$ coincide.
Then $E^{(N)}(x_{1},...x_{N})$ is finite for generic configurations,
or more precisely on the complement of the fixed point locus of the
$S_{N}-$action on $X^{N}$ (but the equality \ref{eq:limit of EN intro}
does not hold as the right hand side is identically $\infty$ if some
$w_{m}$ takes the value $\infty).$ Moreover, it can be shown that
the Assumption $1$ is valid (as discussed below). However, the main
issue is the convexity of the corresponding mean free energy, which
in particular implies that $E^{(N)}(x_{1},...x_{N})$ must be $\lambda-$convex
on the interior of any fundemental domain $\Lambda$ for the $S_{N}-$action
on $X^{N}.$ As it turns out the latter condition is, in fact, also
sufficent for the convexity assumption $2$ two hold when $D=1$ (as
is shown precisely as in the proof of Propostion \ref{prop:hidden convexity}
below). This is the reason that we will mainly consider the one dimensional
setting in Section \ref{sec:Singular-pair-interactions}.

\section{\label{sec:Singular-pair-interactions}Applications to singular pair
interactions in 1D}

In the following it will be convenient to use the following (non-standard)
terminilogy: a continuous function $\psi(x)$ on a convex domain of
$\R^{D}$ is \emph{quasi-convex} if it is $\lambda-$convex, i.e.
it can be written as $\psi(x)=\phi(x)+\lambda|x|^{2}/2$ for some
convex function $\phi$ and if, for $|x|$ sufficently large, $\psi(x)=\phi(x)+o(|x|^{2})$
for some (possibly different) convex function $\phi.$

\subsection{Setup\label{sub:Setup sing}}

Let $w(s)$ be a quasi-convex real-valued function on $]0,\infty[$
such that there exist positive constants $A$ satisfying 
\[
\liminf_{s\rightarrow0}w(s)\geq-A
\]
Extend $w$ to a lsc function $w:\,\R\rightarrow]-\infty,\infty]$
by demanding that $w(-s)=w(s)$ for $s\neq0$ and $w(0):=\liminf_{s\rightarrow0}w(s).$
We define the corresponding \emph{pair interaction function} by 
\[
W(x,y):=w(x-y)(=w(|x-y|)
\]
 which is called \emph{repulsive (attractive)} if $w(s)$ is decreasing
(increasing) on $]0,\infty[.$ Given a quasi-convex function $V(x)$
we define the corresponding $N-$point interaction energy by 
\begin{equation}
E_{W,V}^{(N)}(x_{1},x_{2},....,x_{N}):=\frac{1}{N-1}\frac{1}{2}\sum_{i\neq j}w(x_{i}-x_{j})+V(x_{i})\label{eq:def of interaction energy for pair inter}
\end{equation}

\begin{rem}
Note that even if $\lambda>0$, the function $W(x,y)$ is at best
only $0$-convex due to translation invariance. Since a $\lambda$-convex
function for $\lambda>0$ is also $0$-convex , we will to simplify
notation in the sequel implicitly that assume $\lambda\leq0$.
\end{rem}
We will consider the general setting of an $N-$dependent inverse
temperature $\beta_{N}$ such that 
\[
\lim_{N\rightarrow\infty}\beta_{N}:=\beta\in]0,\infty].
\]
 Then the corresponding SDEs can be formally written as 
\begin{equation}
dx_{i}(t)=-\frac{1}{(N-1)}\sum_{j\neq i}(\nabla w)(x_{i}-x_{j})dt-(\nabla V)(x_{i})dt+\sqrt{\frac{2}{\beta_{N}}}dB_{i}(t),\label{eq:sde for v and w}
\end{equation}
(see Section \ref{sub:The-forward-Kolmogorov}). 
\begin{lem}
\label{lem:approx with cont lambda conv}There exists a sequence of
continuous functions $w_{R}(t)$ on $[0,\infty[$ increasing to $w_{R}$
such that $w_{R}$ is quasi-convex, with a $\lambda$ independent
of $R$ and such that $w_{R}=a_{R}t+o(t^{2})$ for some $a_{R}\in\R$
when $|x|\geq R$\end{lem}
\begin{proof}
The only issue is the $\lambda-$convexity close to $t=0,$ for $t\geq0$
and it may be obtained as follows. First assume that $\lambda=0$
and fix and number $\epsilon>0.$ Let $w^{(\epsilon)}$ be the convex
function on $[0,\infty[$ coinciding with $w$ on the complement of
$]0,\epsilon]$ and on $]\epsilon,\infty[$ it coincides with the
affine function defined by the left tangent line of $w$ at $t=\epsilon.$
By convexity $w^{(\epsilon)}$ increases to $w$ as $\epsilon\rightarrow0.$
In general, if $w$ is $\lambda-$convex for some $\lambda\leq0$
we can decompose $w(t)=:\phi(t)-\lambda t^{2}/2$ where the first
term is convex. Replacing $\phi(t)$ by $\phi^{(\epsilon)}(t)$ as
above and defining $\phi^{(\epsilon)}:=v^{(\epsilon)}(t)-\lambda t^{2}/2$
then gives a sequence of continuos $\lambda-$convex functions on
$[0,\infty[$ increasing to $w.$ Similarly, by decompising $w(t)=\phi(t)+o(t^{2})$
for $|t|\geq R$ for a (possibly different) convex function $\phi$
we can replace $\phi$ on $[R,\infty[$ with the corresponding affine
function which is tangent to $\phi$ as $t$ increases to $R.$ \end{proof}
\begin{rem}
In the case when $w$ is decreasing on $]0,\infty[$ we could simply
take $w_{R}(t):=w(t+1/R)$ above.
\end{rem}
We fix a sequence of quasi-convex continuous functions $w_{R}$ and
$V_{R}$ which are bounded from above and increase to to $w$ and
$V,$ respectively, where $R$ will be referred to as the ``truncation
parameter'' (such sequences exist by the assumption on lower semi-continuity). 
\begin{example}
(power-laws) Our setup applies in particular to the \emph{repulsive}
power-laws 
\[
w(|x|)\sim|x|^{s}\,\,\,s\in]-1,0[,\,
\]
whose role in the case $s=0$ it played by the repulsive logarithmic
potential $w(|x|)=-\log|x|$ and for $s>0$ by 
\[
w(x)\sim-|x|^{s},\,\,\,\,\,s\in]0,1]
\]
The results also apply in the following cases of\emph{ attractive}
potentials (see Section\ref{sub:Convex-interactions-in} for $D>1):$
\[
w(x)\sim|x|^{s}\,\,\,s\in[1,\infty[
\]

\end{example}
Similarly, the assumptions are satisfied by linear combinations of
interactions whose asymptotics as $|x|\rightarrow0$ and $|x|\rightarrow\infty$
are comparible to the power-laws as above. In particular, this is
the case for the interactions used in applications to swarming and
flocking models, which are usually taken to be repulsive at a short
distances and attractive at large distances (see \cite{cch} and references
therein). For example, this is the case for the Morse potential used
in swarming models: 
\[
w(|x|)=C_{R}e^{-|x|/l_{R}}-C_{A}e^{-|x|/l_{A}}
\]
which is clearly $\lambda-$convex on $]0,\infty[$ for some (possibly
negative) $\lambda.$ When $C_{R}/C_{A}>1$ and $l_{R}/l_{A}<1$ it
is repulsive/attractiv at small/large distances.

\subsection{Propagation of chaos in the large $N-$limit and convexity}
\begin{prop}
\label{prop:macroscopic pair energy}The functional 
\begin{equation}
E_{W,V}(\mu):=\frac{1}{2}\int_{\R\times\R}W\mu\otimes\mu+\int V\mu:=\mathcal{W}(\mu)+\mathcal{V}(\mu)\in]\infty,\infty]\label{eq:macro scopic pair energy}
\end{equation}
is well-defined, and lsc on $\mathcal{P}_{2}(\R)$ and satisfies the
coercivity property \ref{eq:coercivity type condition}. Moreover,
\textup{
\begin{equation}
\frac{1}{N}\int E_{W,V}^{(N)}(x_{1},x_{2},....,x_{N})\mu^{\otimes N}=E_{W,V}(\mu),\label{eq:exact formula for mean energy}
\end{equation}
if $E_{W,V}(\mu)<\infty$ and in general} 
\begin{equation}
E_{W,V}(\mu)=\lim_{R\rightarrow\infty}E_{W_{R},V_{R}}(\mu)\label{eq:energy as limit of truncations}
\end{equation}
where $E_{W_{R},V_{R}}(\mu)$ is continuous along any sequence $\mu_{j}$
converging weakly in $\mathcal{P}(\R)$ with a uniform bound on the
second moments. In particular, taking $\mu=\delta_{N}(x_{1},...,x_{N})$
 we have
\[
E_{W_{R},V_{R}}(\delta_{N}(x_{1},...,x_{N}))+O(\frac{C_{R}}{N})=\frac{1}{N}E_{W_{R},V_{R}}^{(N)}(x_{1},x_{2},....,x_{N})
\]
\end{prop}
\begin{proof}
To simplify the notation we assume that $V=0,$ but the general case
is similar. First note that the fact that $E_{W}(\mu)$ is well-defined
is trivial in case $\mu$ has compact support since then $W\geq-C$
on the support of $\mu.$ Formula \ref{eq:exact formula for mean energy}
then follows immediately from the definition, using the Fubini-Tonelli
theorem to interchange the order of integration. In the general case
we note that fixing $\delta>0$ and setting $U_{\delta}:=\{(x,y):\,|x-y|>\delta\}$
gives $\int_{U_{\delta}}W\mu\otimes\mu\geq-C_{\delta}\int|x-y|^{2}\mu\otimes\mu\geq2C_{\delta}\int(|x|^{2}+y|^{2})\mu\otimes\mu:=C'_{\delta}>\infty$
since $\mu\in\mathcal{P}_{2}(\R).$ Hence, $E_{W}(\mu):=\int_{\R\times\R}W\mu\otimes\mu:=\int_{U_{\delta}}W\mu\otimes\mu+\int_{U_{\delta}^{c}}W\mu\otimes\mu$
is well-defined, since $W\geq A_{\delta}$ on $U_{\delta}^{c}.$ The
convergence \ref{eq:energy as limit of truncations} as $R\rightarrow\infty$
then follows from the monotone and dominated convergence theorems.
We note that for any fixed $R$ the functional $E_{W_{R}}(\mu)$ is
continuous along any sequence $\mu_{j}$ converging weakly in $\mathcal{P}(\R)$
with a uniform bound on the second moments, as follows from the fact
that $W_{R}$ is continuous and bounded from above on any compact
subset $[-k,-k]$ of $\R$ together with the simple tail estimate
\begin{equation}
\int_{|x|\geq k}W_{R}\mu_{j}\leq\sup_{|x|\geq k}\frac{W_{R}(x)}{|x|^{2}}\int|x|^{2}\mu_{j},\label{eq:tail estimate}
\end{equation}
 which by the quasi-convexity assumption tends to zero, as $k\rightarrow\infty,$
uniformly in $j.$ Finally, since $E_{W}(\mu)$ is an increasing sequence
of continuous functionals $E_{W_{R}}(\mu)$ it follows that $E_{W}(\mu)$
is lower semi-continuous on $\mathcal{P}_{2}(\R).$ To prove the last
statement we note that, by definition, the error term in question
comes from the missing diagonal terms in the definition of $E_{W_{R}}^{(N)}(x_{1},x_{2},....,x_{N})$
corresponding to $i=j$ i.e. from 
\[
\frac{1}{N}\frac{1}{N-1}\sum_{i=1}^{N}w_{R}(x_{i}-x_{i})=\frac{1}{N}\frac{N}{N-1}w_{R}(0)=O(\frac{C_{R}}{N}),
\]
 which concludes the proof.\end{proof}
\begin{prop}
\label{prop:hidden convexity}(convexity). The mean energy functional
\textup{
\[
\mathcal{E}_{W,V}(\mu^{N}):=\frac{1}{N}\int_{\R^{N}}E_{W,V}^{(N)}\mu^{N},
\]
}restricted to the subspace of symmetric probability measures $\mathcal{P}_{2}(\R^{N})^{S_{N}}$
in $\mathcal{P}_{2}(\R^{N}),$ is $\min(\lambda,0)-$convex along
generalized geodesics with symmetric base $\nu_{N}.$ In particular,
the functional $E_{W,V}(\mu)$ is lsc and $\min(\lambda,0)-$convex
along generalized geodesics in $\mathcal{P}_{2}(\R)$ and satisfies
the coercivity condition \ref{eq:coercivity type condition}. \end{prop}
\begin{proof}
We first claim that it is enough to prove the $\lambda-$convexity
for generalized geodesics in $\mathcal{P}_{2}(\R^{N})_{abs}^{S_{N}}$
such that the base $\nu_{N}$ is also in $\mathcal{P}_{2}(\R^{N})_{abs}^{S_{N}}.$
Indeed, by monotonicity the mean energy functionals corresponding
to the approximations $E_{W_{R},V_{R}}$ $\Gamma-$converge towards
$\mathcal{E}_{W,V}$ and hence we may (by the first point in Prop
\ref{prop:gener conv}) first assume that $E_{W,V}^{(N)}$ is continuous
and bounded with $V$ $\lambda-$convex on $\R$ and $w(t)$ $\lambda-$convex
for $t\geq0.$ Next, we observe that when $E_{W,V}^{(N)}$ continuous
and bounded the corresponding mean energy functional on $\mathcal{P}(\R^{N})$
is continuous wrt the weak topology. Now, given a generalized geodesic
$\mu_{s}$ in $\mathcal{P}_{2}(\R^{N})$ with a base $\nu$ in $\mathcal{P}_{2}(\R^{N})_{abs}$
we can approximate the end points weakly by measures with finite entropy
(by a simple convolution and truncation argument). By the convexity
of the entropy, the corresponding generalized geodesics $\mu_{s}^{(j)}$
have finite entropy for any fixed $s\in[0,1]$ and in particularly
$\mu_{s}^{(j)}$ is a curve in $\mathcal{P}_{2}(\R^{N})_{abs}.$ By
Prop \ref{prop:gener conv}, this proves the claim above. 

We will write $x:=(x_{1},..,x_{N})$ etc. Let $\mu_{0}^{(N)},$ $\mu_{1}^{(N)}$
and $\nu^{(N)}$ be three given symmetric measures in $\mathcal{P}_{2,ac}(\R^{N})$
and denote by $T_{0}$ and $T_{1}$ the optimal maps pushing forward
$\nu^{(N)}$ to $\mu_{0}^{(N)}$ and $\mu_{1}^{(N)},$ respectively.
Let $T_{t}:=(1-t)T_{0}+tT$ so that $\mu_{t}:=T_{t}\nu^{(N)}$ is
the corresponding generalized geodesic. The key point of the proof
is the following 

\emph{Claim:} $(a)$ $T_{t}$ commutes with the $S_{N}-$action and
$(b)$ $T_{t}$ preserves order, i.e. $x_{i}<x_{j}$ iff $T(x)_{i}<T(x)_{j}.$

The first claim $(b)$ follows directly from Kantorovich duality \cite{br,v1}.
Indeed, $T_{i}$ (for $i\in\{0,1\})$ is an optimal transport map
iff we can write $T_{i}=\nabla\phi_{i}$ where the convex function
$\phi_{i}$ on $\R^{N}$ minimizes the Kantorovich functional $J_{i}$
corresponding to the two $S_{N}-$invariant measures $\mu_{i}^{(N)}$
and $\nu^{(N)}.$ But then it follows from general principles that
the minimizer can also be taken $S_{N}-$invariant. To prove the claim
$(b)$ we will use the well-known fact that any optimal map $T$ is
cyclical monotone and in particular for any $x$ and $x'$ in $\R^{N}$
\[
\left|x-T(x)\right|^{2}+|x'-T(x')|^{2}\leq\left|x-T(x')\right|^{2}+|x'-T(x)|^{2}
\]
(as follows from the fact that $T$ is the gradient of a convex function).
In particular, denoting by $\sigma(=(ij))\in S_{N}$ the map on $\R^{N}$
permuting $x_{i}$ and $x_{j}$ we get, 
\[
\left|x-T(x)\right|^{2}+|\sigma x-T(\sigma x)|^{2}\leq\left|x-T(\sigma x)\right|^{2}+|\sigma x-T(x)|^{2}
\]
But since (by $(a))$ $T\sigma=\sigma T$ and $\sigma$ acts as an
isometry on $\R^{N}$ the left hand side above is equal to $2\left|x-T(x)\right|^{2}$
and similarly, since $\sigma^{-1}=\sigma$ the right hand side is
equal to $2|\sigma x-T(x)|^{2}.$ Hence setting $y:=T(x)$ gives 
\[
\left|x-y\right|^{2}\leq|\sigma x-y|^{2}
\]
Finally, expanding the squares above and using that $\sigma=(ij)$
gives $-2(x_{i}y_{i}+x_{j}y_{j})\leq-2(x_{j}y_{i}+x_{i}y_{j})$ or
equivalently: $(x_{i}-x_{j})(y_{i}-y_{j})\geq0,$ which means that
$x_{i}<x_{j}$ iff $y_{i}<y_{j}$ and that concludes the proof of
$(b).$

Now, by the previous claim the map $T_{t}$ preserves the fundamental
domain
\begin{equation}
\Lambda:=\{x:\,x_{1}<x_{2}<...<x_{N}\}\label{eq:fundam domain}
\end{equation}
 for the $S_{N}-$action on $\R^{N}.$ But, by assumption, on the
subset $\Lambda$ the function $E^{(N)}$ is convex and this is enough
to run the usual argument to get convexity of the mean energy on the
subspace of symmetric measures. Indeed, we can decompose 
\begin{equation}
\int_{X^{N}}E^{(N)}\mu_{t}^{(N)}=\sum_{\sigma}\int_{\sigma(\Lambda)}E^{(N)}(T_{t})_{*}\nu^{(N)}\label{eq:decomp of mean energy}
\end{equation}
(using that $(T_{t})_{*}\nu^{(N)}$ does not charge null sets, since
it has a density). For any fixed $\sigma$ the integral above is equal
to $\int_{\sigma(\Lambda)}T_{t}^{*}E^{(N)}\nu^{(N)}$ (since $T_{t}$
preserves $\sigma(\Lambda))$ which, by the $S_{N}-$invariance of
$\nu^{(N)}$ and $E^{(N)}$ in turn is equal to $\int_{\Lambda}T_{t}^{*}E^{(N)}\nu^{(N)}.$
But since $E^{(N)}$ is convex on $\Lambda$ and $T_{t}$ preserves
$\Lambda$ the function $T_{t}^{*}E^{(N)}$ is convex in $t$ for
any fixed $x\in\Lambda$ and hence, by the decomposition \ref{eq:decomp of mean energy},
$\int_{X^{N}}E^{(N)}\mu_{t}^{(N)}$ is convex wrt $t,$ as desired.
Finally, the convexity of $E(\mu)$ follows immediately by taking
$\mu^{(N)}$ to be a product measure $\mu^{\otimes N}$ and using
formula \ref{eq:exact formula for mean energy}.\end{proof}
\begin{rem}
\label{rem:hidden}The first convexity statement may appear to contradict
the second point in Lemma \ref{lem:generalized conv}, which seems
to force $E^{(N)}$ to be convex on all of $\R^{N}$(which will not
be the case in general). But the point is that we are only integrating
against \emph{symmetric} measures. As for the convexity of $E_{W,V}(\mu)$
is is indeed well-known that it holds precisely when the symmetric
function $w(x)$ is convex on $]0,\infty[$ (see \cite{bl,Carrillo-f-p}).
This fact can be proved more directly by using that, in this special
case, $T(x)=(f(x_{1}),f(x_{2}),....f(x_{N}))$ clearly preserves order
since $f,$ being the derivative of a convex function, is clearly
increasing.
\end{rem}
Given a sequence $\beta_{N}\in]0,\infty]$ converging to $\beta\in]0,\infty]$
we recall that $\mathcal{F}^{(N)}$ denotes the corresponding mean
free energy functional \ref{eq:def of mean free energy}. Similarly,
we define the the corresponding\emph{ (macroscopic) free energy functional}
on $\mathcal{P}(\R)$ by 
\begin{equation}
F_{\beta}(\mu):=E_{W,V}(\mu)+\frac{1}{\beta}H(\mu)\label{eq:free energy functional on prob}
\end{equation}
Combining the previous proposition with Theorem \ref{thm:existence of evi}
shows that the EVI-gradient flows $\mu_{t}$ of $F_{\beta_{N}}$ and
$F_{\beta}^{(N)}$ on $\mathcal{P}_{2}(\R)$ and $\mathcal{P}_{2}(\R^{N})^{S_{N}},$
respectively exist for appropriate initial measures.
\begin{thm}
\label{thm:w}Let $W$ and $V$ be a two-point interaction energy
and potential as in Section \ref{sub:Setup sing} and denote by $\mu_{t}^{(N)}$
the corresponding probability measures on $\R^{n}$ evolving according
to the forward Kolmogorov equation associated to the stochastic process
\ref{eq:sde for v and w}. Assume that at the initial time $t=0$
\[
\lim_{N\rightarrow\infty}(\delta_{N})_{*}\mu_{t}^{(N)}=\delta_{\mu_{0}}
\]
in the $L^{2}-$Wasserstein metric. Then, at any positive time 
\[
\lim_{N\rightarrow\infty}(\delta_{N})_{*}\mu_{t}^{(N)}=\delta_{\mu_{t}}
\]
in the $L^{2}-$Wasserstein metric, where $\mu_{t}$ is the EVI-gradient
flow on $\mathcal{P}_{2}(\R)$ of the free energy functional $F_{\beta},$
emanating from $\mu_{0}.$ \end{thm}
\begin{proof}
Assumptions $2$ and $3$ in Section \ref{sec:A-general-convergence}
have been verified above and we just need the verify Assumption $1$
in order to apply Theorem \ref{thm:general}. The upper bound \ref{eq:upper bound in induction}
follows precisely as before, using formula \ref{eq:exact formula for mean energy}.
In order to verify the lowerer bound we fix the truncation parameter
$R>0$ and observe that, since, $E_{W,V}^{(N)}\geq E_{W_{R},V_{R}}^{(N)}$,
formula \ref{eq:energy as limit of truncations} gives 
\[
E_{W,V}^{(N)}(\mu_{N})/N\geq\int E_{W_{R},V_{R}}(\delta_{N}(x_{1},...,x_{N}))\mu_{N}+C_{R}/N
\]
But 
\[
\int E_{W_{R},V_{R}}(\delta_{N}(x_{1},...,x_{N}))\mu_{N}=\int_{\mathcal{P}(\R)}E_{W_{R},V_{R}}(\mu)(\delta_{N})_{*}\mu_{N}\rightarrow\int_{\mathcal{P}(\R)}E_{W_{R},V_{R}}(\mu)\Gamma
\]
as $N\rightarrow\infty,$ by the continuity properties of the functional
$E_{W_{R},V_{R}}(\mu)$ (Prop \ref{lem:approx with cont lambda conv}).
Hence, 
\[
\liminf_{N\rightarrow\infty}\frac{1}{N}\int_{\R^{N}}E^{(N)}\mu^{(N)}\geq\int E_{W_{R},V_{R}}(\mu)\Gamma
\]
for any $R>0.$ Finally, letting $R\rightarrow\infty$ and using the
monotone convergence theorem concludes the proof. 
\end{proof}

\subsection{\label{sub:Deterministic-limits-for}Deterministic mean field limits
for strongly singular pair interactions}

Let us next specialize to deterministic ``strongly singular'' interactions,
i.e. the case when $\beta_{N}=\infty$ and $w(s)$ blows up as $s\rightarrow0.$
In this case the corresponding EVI gradient flows on $\R^{N}$ in
the deterministic case $\beta_{N}=\infty,$ are induced by the classical
solutions to the corresponding system of ODEs 
\begin{lem}
\label{lem:strong sing}Fix a positive integer $N>0.$ Let $w$ be
as in Section \ref{sub:Setup sing} and assume moreover that $w$
is $C^{2}-$smooth on $]0,\infty[$ with $w(s)\rightarrow\infty$
as $s\rightarrow0$ and that $|\nabla w|$ remains bounded as $s\rightarrow\infty.$
Then 
\begin{itemize}
\item there is a unique smooth solution to the corresponding ODE on $\R^{N}$
obtained by setting $\beta_{N}=\infty$ in equation \ref{eq:sde for v and w}:
\[
dx_{i}(t)/dt=-\sum_{j\neq i}(\nabla w)(x_{i}-x_{j})
\]
if the initial condition satisfies $x_{i}(0)\neq x_{i}(0)$ when $i\neq j$
(then this condition is preserved for any $t>0.$
\item the corresponding curve of symmetric probability measures 
\begin{equation}
\mu^{(N)}(x_{1}(t),....,x_{N}(t)):=\frac{1}{N!}\sum_{\sigma\in S_{N}}\delta_{(x_{\sigma(1)}(t),....,x_{\sigma(N)}(t))}\label{eq:def of symmetric dirac prob measure}
\end{equation}
on $\R^{N}$ (i.e. the normalized Dirac measure supported on the $S_{N}-$orbit
in $\R^{N}$ of $(x_{1}(t),....,x_{N}(t))$ coincides with the unique
EVI gradient flow solution to the corresponding mean energy functional
$\mathcal{E}^{(N)}$ on $\mathcal{P}(\R^{N})^{S_{N}},$ with initial
data $\mu^{(N)}(x_{1}(0),....,x_{N}(0)).$
\item if the initial coordinates $x_{i}(0)$ are viewed as $N$ iid random
variables on $\R$ with distribution $\mu_{0}\in\mathcal{P}_{2}(\R),$
then the law $(x_{1}(t),....,x_{N}(t))_{*}\mu_{0}^{\otimes N}$ of
the corresponding random variable \textup{$x_{1}(t),....,x_{N}(t)),$
for a given positive time $t,$ gives the unique EVI gradient flow
solution of the }corresponding mean energy functional $\mathcal{E}$
on $\mathcal{P}(\R^{N})$ with initial data $\mu_{0}^{\otimes N}.$
\end{itemize}
\end{lem}
\begin{proof}
Set $x:=(x_{1},...,x_{N}),$ where we may without loss of generality
assume that $x_{1}<...<x_{N},$ i.e. that $x$ is a point in the fundamental
domain $\Lambda$ for the $S_{N}-$action. By the standard Cauchy-Lipschitz
existence result for ODE's with Lipschitz continuous drift there exists
$T>0$ and a solution $x(t)$ for $t\in[0,T[$ where $x(t)$ stays
in the open convex set $\Lambda.$ Moreover, since the flow $x_{t}$
on $\Lambda$ is the gradient flow (wrt the Euclidean metric) of the
smooth function $E^{(N)}(:=E_{W}^{(N)})$ on $\Lambda$ it follows
that 
\begin{equation}
E^{(N)}(x(t))\leq E^{(N)}(x(0))=:A<\infty.\label{eq:bound A}
\end{equation}
 Moreover, we claim that
\begin{equation}
\left\Vert x(t)\right\Vert \leq B\label{eq:bound B}
\end{equation}
 for some constant $B$ only depending on the initial data and $T.$
Accepting this for the moment it then follows from the bounds \ref{eq:bound A},
\ref{eq:bound B} and the singularity assumption on $w(0)$ that there
is a positive constant $\delta$ such that $|x_{i}(t)-x_{j}(t)|\geq\delta$
for any $(i,j)$ such that $i\neq j$ and any $t\in[0,T[.$ Hence,
by restarting the flow again the short-time existence result translates
into a long-time existence result, i.e. $T=\infty.$ Before establishing
the claimed bound \ref{eq:bound B} it may be worth pointing out that
the bound in question is not needed in the case when $w(s)\geq-C$
as $s\rightarrow\infty$ (indeed, the $+\infty$ singularity appearing
when two particles merge can then not be compensated by letting some
particles go off to infinity), but this problem could a priori appear
if, for example, $w(s)=-\log s.$ To prove the bound \ref{eq:bound B}
we observe that it will be enough to prove the second point in the
statement of the lemma. Indeed, since the second moments of the EVI
gradient flow $\mu^{(N)}(x_{1}(t),....,x_{N}(t))$ are uniformly bounded
on any fixed time interval this will imply the desired uniform bound
\ref{eq:bound B} on $\R^{N}.$

We next turn to the proof of the second point. We recall that by the
convexity result in Proposition \ref{prop:hidden convexity} together
with the general result Theorem \ref{thm:existence of evi} the functional
$\mathcal{E}$ admits an EVI gradient flow $\mu^{(N)}(t)$ on $\mathcal{P}(\R^{N})_{2}$
emanating from any given symmetric measure $\mu^{(N)}(0)\in\mathcal{P}(\R^{N})_{2}^{S^{N}}.$
Here we will take $\mu^{(N)}(0)$ to be as in formula \ref{eq:def of symmetric dirac prob measure}.
Setting $Y=\R^{N}/S^{N}$ we can identify the space $\mathcal{P}(\R^{N})_{2}^{S_{N}}$
with $\mathcal{P}(Y)_{2}$ and $\mu^{(N)}(t)$ with an EVI gradient
flow $[\mu^{(N)}(t)]$ on $\mathcal{P}(Y)_{2}$ emanating from $[\mu^{(N)}(t)]=\delta_{y},$
where $y$ is the point $[x(0),...x(N)]\in Y.$ We claim that $[\mu^{(N)}(t)]$
is of the form $\delta_{y(t)}$ for a curve $y(t)$ in $\mathcal{P}(Y)_{2}.$
Indeed, since $[\mu^{(N)}(t)]$ arises as limit of minimizing movements
it is enough to establish the claim for the minimizing movement corresponding
to any given time discretization. But in the latter situation the
corresponding functionals 
\[
J_{j+1}(\cdot):=\frac{\text{1}}{2\tau}d(\cdot,\mu_{t_{j}})^{2}+\mathcal{E}(\cdot)-\mathcal{E}(\mu_{t_{j}})
\]
 appearing in formula \ref{eq:def of J functional} (with $F=\mathcal{E}$)
are linear (wrt the ordinary affine structure) at each time step if
$\mu_{t_{j}}$ is assumed to be a Dirac mass, i.e of the form $\delta_{y_{t_{j}}}$
(since $\mathcal{E}$ and $d^{2}(\cdot,\delta_{y_{t_{j}}})$ are both
linear). Hence, by Choquet's theorem, any optimizer is of the form
$\delta_{y_{t_{j+1}}}$ for some point $y_{t_{j+1}}\in Y$ and thus,
by induction, this proves the claim. Next we apply the general isometric
embedding 
\begin{equation}
Y\rightarrow\mathcal{P}_{2}(Y),\,\,\,y\mapsto\delta_{y}\label{eq:isometr embed proof lemma ode}
\end{equation}
 which clearly has the property that $E(y)=E(x^{(N)})=\mathcal{E}(\delta_{y}).$
The EVI gradient flow $[\delta_{y(t)}]$ on $\mathcal{P}(Y)_{2}$
thus gives rise to an EVI gradient flow $y(t)$ on $Y.$ Finally,
the curve $y(t)\in\R^{N}/S_{N}$ may be identified with a curve $x(t)$
in the domain $\Lambda\subset\R^{N},$ which concludes the proof of
the second point. 

The last point can be proved in a similar manner by approximating
the initial measure with a sum of Dirac masses. Alternatively, one
can use that the result is well-known when $E$ is smooth on $\R^{N}$
and the initial symmetric measure $\nu_{0}$ on $\R^{N}$ has a smooth
density (indeed, then $y(t)_{*}\nu_{0}$ satisfies a transport equation
which defines the unique EVI gradient flow solution of $\mathcal{E}$
emanating from $\nu_{0}).$ The general case then follows by approximation
using the stability property of EVI gradient flows \cite{a-g-s}.\end{proof}
\begin{rem}
By Theorem \ref{thm:existence of evi} one can, in fact, start the
EVI gradient flow of $\mathcal{E}^{(N)}$ from any symmetric measure
$\mu^{(N)}(0)$ and in particular from any measure of the form $\mu^{(N)}(x_{1}(0),....,x_{N}(0)),$
given an arbitrary configuration $x_{1}(0),....,x_{N}(0).$ Since,
$\mathcal{E}^{(N)}(\mu^{(N)}(t)<\infty$ when $t>0$ this corresponds
to a classical solution $(x_{1}(t),....,x_{N}(t))$ of mutually distinct
points when $t>0.$ 
\end{rem}
We can now give a purely deterministic version of Theorem \ref{thm:w}.
We also establish the corresponding propagation of chaos result when
the initial particle positions are taken as random iid variables:
\begin{thm}
\label{thm:deterministic strong}Assume that $w$ defines a strongly
singular two-point interaction and denote by $x^{(N)}(t):=(x_{1},...,x_{N})(t)$
the solution of the system of ODEs on $\R^{N}$ in the previous lemma. 
\begin{itemize}
\item If 
\[
\frac{1}{N}\sum_{i=1}^{N}\delta_{x_{i}(0)}\rightarrow\mu_{0}
\]
in $\mathcal{P}_{2}(\R)$ for a given measure $\mu_{0}$ in $\mathcal{P}_{2}(\R).$
Then, for any positive time $t,$ 
\[
\frac{1}{N}\sum_{i=1}^{N}\delta_{x_{i}(t)}\rightarrow\mu_{t},
\]
 where $\mu_{t}$ is the EVI gradient flow solution of the corresponding
functional $E_{W}$ on $\mathcal{P}_{2}(\R).$
\item Similarly, if the $x_{i}(0):$s are viewed as $N$ iid random variables
with distribution $\mu_{0}\in\mathcal{P}_{2}(\R),$ then the corresponding
random measure $\frac{1}{N}\sum_{i=1}^{N}\delta_{x_{i}(t)}$ converges
to $\mu_{t}$ in law, as $N\rightarrow\infty$ (and hence propagation
of chaos holds). 
\end{itemize}
\end{thm}
\begin{proof}
To prove the first point we continue with the notation in the previous
lemma and note that
\begin{equation}
\Gamma_{N}:=(\delta_{N})_{*}\mu^{(N)}(x^{(N)})=\delta_{\delta_{x^{(N)}}},\label{eq:proof det}
\end{equation}
i.e. $\Gamma_{N}$ is the delta measure on $\mathcal{P}(\mathcal{P}(\R^{N}))$
supported on the measure $\delta_{x^{(N)}}$ (as follows directly
from the definitions). By assumption $\Gamma_{N}(0)\rightarrow\delta_{\mu_{0}}$
in $\delta_{x^{(N)}}$ in $\mathcal{P}_{2}(\mathcal{P}_{2}(\R^{N}))$
as $N\rightarrow\infty.$ But then it follows, precisely as in the
proof of Theorem \ref{thm:w}, that $\Gamma_{N}(t)\rightarrow\delta_{\mu_{t}}$
in $\mathcal{P}_{2}(\mathcal{P}_{2}(\R^{N}))$ for any fixed $t>0.$
Finally, using the relation \ref{eq:proof det} again concludes the
proof of the first point in the theorem, using the very definition
of $\Gamma_{N}(t)$ (and the $S_{N}-$invariance of the empirical
measure $\delta_{N}).$ As for the last point it follows immediately
from the last point in the previous lemma combined with Theorem \ref{thm:w}
(applied to the case $\beta_{N}=\infty).$
\end{proof}

\subsection{\label{sub:An-alternative-approach}An alternative approach not involving
the $S_{N}-$action }

In this section we point out that when $D=1$ it is possible to dispense
with the $S_{N}-$action and thus bypass the use of the $G-$equivariance
in Theorem \ref{thm:existence of evi}. On the other hand, one virtue
of the more geometric approach used above is that it may turn out
to be useful when $D>1.$ Indeed, given an interaction energy $E^{(N)}$
the Assumption $2$ in Section \ref{sub:The-assumptions-on} can be
replaced by the convexity of $E^{(N)}$ along any class of ``interpolating
curves'' satisfying the two conditions appearing in the proof of
Theorem \ref{thm:existence of evi} (as emphasized in the introduction
of \cite{a-g-s} this is all that it is needed to develop the general
theory of gradient flows in \cite{a-g-s}). However, finding such
interpolating curves for a given singular interaction $E^{(N)}$ appears
to be a highly non-trivial task. 

Denote by $\Lambda_{N}\subset\R^{N}$ the interior of the standard
fundamental domain for the $S_{N}-$action on $\R^{N}.$ Its closure
$\overline{\Lambda_{N}}$ is the closed convex subset $\{x_{1}\leq x_{2}\leq....\leq x_{N}\}$
that inherits a (geodesically) complete metric from the Euclidean
metric on $\R^{N}.$ The starting point is the following basic
\begin{lem}
The quotient map induces an isometry $\Phi$ between the the Euclidean
space $\overline{\Lambda_{N}}$ and the metric quotient space $\R^{N}/S^{N}.$ \end{lem}
\begin{proof}
For completeness we give the simple proof. The quotient map induces
a surjective map $\Phi$ from $\overline{\Lambda_{N}}$ to $\R^{N}$
and it is, by continuiuty, enough to prove that $\Phi$ is an isometry
on the dense open subset $\Lambda_{N}$ of $\overline{\Lambda_{N}}.$
But, by definition, $d(\Phi(x_{1},...,x_{N}),\Phi(y_{1},...,y_{N})^{2}=\inf_{\sigma\in S_{N}}\sum_{i}|x_{i}-y_{\sigma(i)}|^{2},$
which, by monotonicity, is realized when $\sigma$ is the identity,
as desired. 
\end{proof}
Using the isometry in the previous lemma we can identify the space
$\mathcal{P}_{2}(\R^{N})^{S^{N}}$ with the space of all probability
measures on $\overline{\Lambda_{N}}$. Concretely, this means that
a measure $\mu_{N}$ on $\overline{\Lambda_{N}}$ is identified with
the avarage of its $S_{N}-$orbit in $\R^{N}.$ In this way the EVI
gradient flow $\mu_{N}(t)$ of the mean energy functional $\mathcal{E}^{(N)}$
on $\mathcal{P}_{2}(\R^{N})^{S^{N}}$ considered above can be identified
with the EVI gradient flow of a functional $\widetilde{\mathcal{E}^{(N)}}$
on $\mathcal{P}_{2}(\R^{N}),$ emanating from a measure supported
on $\overline{\Lambda_{N}},$ where $\widetilde{\mathcal{E}^{(N)}}$
is the linear functional associated to the function $\widetilde{E^{(N)}}:=\chi_{\overline{\Lambda_{N}}}+E^{(N)},$
where $\chi_{\overline{\Lambda_{N}}}$ is the indicator function of
$\overline{\Lambda_{N}},$ i.e. equal to zero on $\overline{\Lambda_{N}}$
and infinity on its complement. By our assumptions $\widetilde{E^{(N)}}$
is $\lambda-$convex on all of $\R^{N}$ and hence, by Prop \ref{prop:macroscopic pair energy},
$\widetilde{\mathcal{E}^{(N)}}$ satisfies the assumptions in Theorem
\ref{thm:existence of evi}, with $G$ trivial. 
\begin{rem}
\label{rem:monotone maps}In the deterministic setting $\beta_{N}=\infty$
this means that one can identify the EVI gradient flow of $E^{(N)}$
on $\R^{N}/S^{N}$ with the gradient flow of $\widetilde{E^{(N)}}$
the Euclidean space$\R^{N},$ defined in terms of the classical theory
of gradient flows on Hilbert spaces (the multivalued subdifferential
of $\widetilde{E^{(N)}}$ at the boundary of $\overline{\Lambda_{N}}$
then gets a contribution from the normal cone of the boundary of $\overline{\Lambda_{N}}).$
Similarly, the identifications above also show that when $\beta_{N}<\infty$
the EVI gradient flow of the corresponding mean energy functional
coincides with the laws of the strong solution of the corresponding
SDEs constructed in \cite{ce} using the theory mulivalued montone
maps (see also \cite{a-g-z}). 
\end{rem}

\section{\label{sub:Realization-of-}Realization of $\mu_{t}$ as a weak solution
of the McKean-Vlasov equation}

In this section we consider the relations between the EVI gradient
flow $\mu_{t}$ appearing in Section \ref{sec:Singular-pair-interactions}
and weak solutions of the corresponding McKean-Vlasov equation, generalizing
some results in \cite{a-g-s,Carrillo-f-p} to the present setting
(and bypassing a gap in the argument in \cite{Carrillo-f-p} in the
case $\beta=\infty$). We conclude by briefly pointing out some subtle
regularity problems of the flows in question. First recall that, in
a general metric space $(M,d)$ the metric slope $|dF|(\mu)$ of a
functional $F$ on $M$ at $\mu\in M$ such that $F(\mu)<\infty$
is defined by 
\[
\left|\partial F\right|(\mu):=\limsup_{\nu\rightarrow\mu}\frac{(F(\nu)-F(\mu))^{+}}{d(\mu,\nu)}
\]
and if $F(\mu)=\infty,$ then $dF|(\mu):=\infty.$ As explained in
\cite{a-g-s} the subdifferential calculus on the Wasserstein $\mathcal{P}_{2}(\R^{D})$
is considerably simpler in the case when $F$ satisfies the following
assumption:
\begin{equation}
|dF|(\mu)<\infty\implies\mu\in\mathcal{P}_{abs}(\R^{D})\label{eq:assump finite slope implies abs}
\end{equation}
(which we will assume below). In particular, this is always the case
for $F=F_{\beta}$ with $\beta<\infty,$ where $F_{\beta}$ denotes
the free energy functional \ref{eq:def of free energy of v notation}.
Now the \emph{subdifferential} $(\partial F)(\mu)$ at $\mu$ of a
$\lambda-$convex functional on $\mathcal{P}_{2}(\R^{D})$ satisfying
the assumption \ref{eq:assump finite slope implies abs} may be defined
as the convex subset of all $\xi\in L^{2}(\R^{d},\mu)$such that 
\[
F(\nu)\geq F(\mu)+\left\langle \xi,\nabla\phi_{\mu.\nu}\right\rangle _{L^{2}(\R^{D},\mu)}+\frac{\lambda}{2}d_{2}(\mu,\nu)^{2},\,\,\forall\nu\in\mathcal{P}_{2}(\R^{D})
\]
where $T_{\mu,\nu}:=I+\nabla\phi_{\mu.\nu}$ is the unique optimal
$L_{loc}^{\infty}-$map transporting $\mu$ to $\nu$ (see \cite[10.1.1 B]{a-g-s}).
The corresponding \emph{minimal subdifferential} $(\partial F)^{0}(\mu)$
is defined as the unique element in $(\partial F)(\mu)\subset L^{2}(\R^{d},\mu)$
with minimal norm.
\begin{thm}
\cite{a-g-s}\label{thm:evi is subdiff gradient flow} Let $F$ be
a lsc real-valued functional on $\mathcal{P}_{2}(\R^{d})$ which is
$\lambda-$convex along generalized geodesics and satisfies the coercivity
property \ref{eq:coercivity type condition}. Then the corresponding
EVI gradient flow $\mu_{t}$ emanating from a given $\mu_{0}$ has
the property that $|dF|(\mu)<\infty$ for $t>0$ and $\mu_{t}$ is
a weak solution on $\R^{d}\times]0,\infty[$ of the continuity equation
\begin{equation}
\frac{d}{dt}\mu_{t}=-\nabla\cdot(\mu_{t}v_{t}),\label{eq:cont eq}
\end{equation}
 where the time-dependent Borel vector fields $v_{t}$ with the property
that $v_{t}=-(\partial^{0}F)(\mu_{t})$ for a.e. $t>0.$\end{thm}
\begin{rem}
The result applies without the assumption \ref{eq:assump finite slope implies abs}
along $\mu_{t},$ but with a more elaborate definition of $\partial^{0}F$
(by \cite[Thm 11.1.3, Thm 11.2.1]{a-g-s}. Briefly, one first defines
the extended subdifferential $(\boldsymbol{\partial}F)(\mu)$ consisting
of transport plans and then defines $(\partial F)(\mu)$ as the subset
realized by transport maps (thus corresponding to vector fields) \cite[10.3.12]{a-g-s}.
In this general setting the previous theorem then says that for a.e.
$t$ the minimal transport plan is realized by the transpart map defined
by the vector field $v_{t}.$ 
\end{rem}
In order to apply the general theory above we will need the following
\begin{lem}
Assume that $V$ and $W$ are $\lambda-$convex and $C^{1}$ on $\R$
and $]0,\infty[,$ respectively. Given $\mu\in\mathcal{P}_{2}(\R)$
such that $|dE_{W,V}|(\mu)<\infty$ and $\mu\in\mathcal{P}_{abs}(\R)$
the minimal subdifferential $\omega_{\mu}:=(\partial^{0}E_{W,V})(\mu)\in L^{2}(\R,\mu)$
satisfies 
\begin{equation}
\int\omega_{\mu}\psi\mu=\frac{1}{2}\int W'(x-y)(\psi(x)-\psi(y))\mu(x)\otimes\mu(y)+\int V'\psi\mu(<\infty)\label{eq:min-sub-dif}
\end{equation}
 for any $\psi$ in $\omega_{\mu}.$\end{lem}
\begin{proof}
This is proved essentially as in \cite[Lemma 3.7]{Carrillo-f-p} (see
also \cite[Thm 10.4.11]{a-g-s} for the the general higher dimensional
case, under the stronger assumption that when $w$ is differentiable
on all of $\R^{D}).$ But for completeness we recall the proof. First
note that using \cite[Prop 10.4.2]{a-g-s} (concerning the case $W=0)$
we may as well assume, by linearity, that $V=0$ and by replacing
$w(x)$ with $w(x)-\lambda|x|^{2}/2$ we may as well assume that $w$
is convex (using that, by the Cauchy-Schwartz inequality $x\psi\in L^{2}(\mu)).$
Now, assume that $\psi\in L_{loc}^{\infty}(\R)\cap L^{2}(\mu)$ and
consider the family of $L^{\infty}-$maps $T_{t}(x):=x+t\psi(x)$
on $\R$ for $t\geq0.$ On one hand a direct calculation gives for
$\mathcal{W}:=E_{W,0}$ that 
\begin{equation}
\lim_{t\rightarrow0^{+}}(\frac{\mathcal{W}((T_{t})_{*}\mu)-\mathcal{W}(\mu)}{t})=\frac{1}{2}\int W'(x-y)(\psi(x)-\psi(y))\mu(x)\otimes\mu(y)\label{eq:min-subdiff-proof}
\end{equation}
 using that, by convexity, $\left(w\left((x-y)+t((\psi(x)-\psi(y))\right)-w(x-y)\right)/t$
is definied and nondecreasing in $t$ for a.e. $(x,y)$ such that
$x>y$ and also using that $W'$ is odd on $\R-\{0\}.$ Indeed, applying
the monotone convergence theorem then gives the previous equality,
since we have assumed that $\mu\in\mathcal{P}_{abs}(\R)$. On the
other hand it follows directly from the definition of the metric slope
that 
\[
\left|\lim_{t\rightarrow0^{+}}(\frac{\mathcal{W}((T_{t})_{*}\mu)-\mathcal{W}(\mu)}{t})\right|\leq|d\mathcal{W}|(\mu)\left\Vert \psi\right\Vert _{L^{2}(\R,\mu)}
\]
But, since the RHS in formula \ref{eq:min-subdiff-proof} is linear
wrt $\psi$ it then follows from the Riesz representation theorem
that there exists a unique element $\omega_{\mu}\in L^{2}(\R,\mu)$
satisfying formula \ref{eq:min-sub-dif} with $\left\Vert \omega_{\mu}\right\Vert _{L^{2}(\R,\mu)}\leq|d\mathcal{W}|(\mu).$
Finally, to verify that $\omega_{\mu}\in\partial\mathcal{W}$ we fix
$\nu\in\mathcal{P}(\R).$ Since $\mu\in\mathcal{P}_{abs}(\R)$ there
exists a unique transport map $T:=T_{\mu,\nu}:=I+\nabla\phi_{\mu.\nu}$
such that $T_{*}\mu=\nu.$ Setting $\psi:=\nabla\phi_{\mu,\nu}$ and
using the convexity of $\mathcal{W}$ on the Wasserstein space thus
gives 
\[
\mathcal{W}(\nu)-\mathcal{W}(\mu)\geq\lim_{t\rightarrow0^{+}}(\frac{\mathcal{W}((T_{t})_{*}\mu)-\mathcal{W}(\mu)}{t})=\left\langle \omega_{\mu},\nabla\phi_{\mu.\nu}\right\rangle _{L^{2}(\R^{D},\mu)}
\]
showing that $\omega_{\mu}\in\partial\mathcal{W},$ as desired.
\end{proof}
As a consequence, when $\beta>0,$ the limiting curve $\mu_{t}$ appearing
in Theorem \ref{thm:general pair intro} is a weak solution of the
McKean -Vlasov equation \ref{eq:drift-d equa intro} in the following
sense 
\begin{prop}
\label{prop:weak solution}Let $V$ be a $\lambda-$convex function
on $\R$ and let $w$ be as in the previous proposition and denote
by $\mu_{\beta}(t)$ the EVI gradient flow $\mu_{\beta}(t)$ $\mathcal{P}_{2}(\R)$
of the corresponding free energy functional $F_{\beta}$ (formula
\ref{eq:def of free energy of v notation}) emanating from a given
$\mu(0)\in\mathcal{P}_{2}(\R).$ Then, given any $\phi\in C_{c}^{2}(\R)$
$\mu_{\beta}(t)$ is a distributional solution of the following equation
on $]0,\infty[:$ 
\[
\frac{d}{dt}\int_{\R}\mu_{\beta}(t)\phi(x)=
\]
\[
=\frac{1}{\beta}\int\mu_{\beta}(t)\phi''(x)+\frac{1}{2}\int W'(x-y)\left(\phi'(x)-\phi'(y)\right)\mu_{\beta}(t)\otimes\mu_{\beta}(t)+\int V'(x)\phi(x)\mu_{\beta}(t)
\]
\end{prop}
\begin{proof}
Since $F_{\beta}(\mu_{\beta})<\infty,$ for $t>0$ and $\beta<\infty$
the assumption \ref{eq:assump finite slope implies abs} holds along
the curve $\mu_{\beta}(t),$ when $t>0$ and hence it follows from
\cite[Thm 11.1.3, Thm 11.2.1]{a-g-s} that the density $\rho_{\beta}(t)$
of $\mu_{\beta}(t)$ is a weak solution on $\R\times]0,\infty[$ of
the following equation 
\[
\frac{d}{dt}\rho_{\beta}(t)=\frac{1}{\beta}\Delta\rho_{\beta}(t)+\nabla(\rho_{\beta}(t)v_{t})),
\]
 where $v_{t}$ is the curve of Borel vector fields defined by $v_{t}(x)=-\partial^{0}E_{W,V}(\mu_{t}).$
The result then follows from the previous proposition.
\end{proof}
In order to consider the case $\beta=\infty$ we will use the following
general stability result: 
\begin{lem}
\label{lem:stability-wrt-beta}(stability wrt $\beta).$ Let $\mu_{\beta}(t)$
the EVI gradient flows in the previous proposition, emanating from
a given $\mu(0)\in\mathcal{P}_{2}(\R),$ independent of $\beta.$
Furhter assume that $E(\mu)<\infty$ for all compactly supperted measures
with an $L^{\infty}$ density. Then $\mu_{\beta}(t)\rightarrow\mu_{\infty}(t)$
in $\mathcal{P}_{2}(\R),$ as $\beta\rightarrow\infty.$\end{lem}
\begin{proof}
According to \cite[Thm 11.2.1]{a-g-s} we just have to verify that
$F_{\beta}$ gamma-converges towards $F_{\infty}(=E)$ on $\mathcal{P}_{2}(\R).$
First observe that in order to verify the $\liminf$ inequality appearing
in the definition of Gamma convergence it is, by a standard diagonal
argument, enough to verify it for the dense subset of all elements
in $\mathcal{P}_{2}(\R)$ with finite entropy. But for any such element
$\mu$, $F_{\beta}(\mu)\rightarrow F_{\infty}(\mu)$ trivially. 

We next turn to the verification of the limitsup inequality in the
definition of Gamma convergence. Note that if $E(\mu)=\infty$ or
$H(\mu)<\infty$, there is nothing to show. Further, since gamma-convergence
is stable under continuous perturbations, we may as well assume that
$E(\mu)$ is convex. So assume that $E(\mu)<\infty$ and $H(\mu)=\infty$.
Then we need to find a recovery sequence $\mu_{\beta}\rightarrow\mu$
such that $H(\mu_{\beta})/\beta\rightarrow0$ and $E(\mu_{\beta})\rightarrow E(\mu)$.
To this end let $\mu_{0}=dx|_{[0,1]}$ be the Lebesgue measure on
the unit interval, and let $\mu_{1}=\mu.$For $t\in[0,1]$ we let
$\mu_{t}$ be the displacement interpolation. Then $E(\mu_{t})$ is
a convex function of $t$, and it follows that $E(\mu)\geq\limsup_{t\to1}E(\mu_{t})$
since $E(\mu_{0})<\infty$. Hence, it is enough to show that for some
choice $t=t(\beta)$, $\lim_{\beta\to\infty}t(\beta)=1$, it holds
that $H(\mu_{t(\beta)})/\beta\to0$ as $\beta\to\infty$. We claim
that $H(\mu_{t})<\infty$ for all $t\in[0,1[$ , from which the result
is immediate. To show the claim, note that since $\mu_{0}$ is supported
on a convex set, there is a unique Brenier map $\partial\phi_{1}$
transporting $\mu_{0}$ to $\mu_{1}$, and it holds that the displacement
interpolation is given by $\phi_{t}=(1-t)x^{2}/2+t\phi_{1}$. But
then $MA(\phi_{t})=\partial^{2}\phi_{t}=(1-t)+t\partial^{2}\phi_{1}\geq1-t$.
Thus an $L^{\infty}$ bound follows for the density $\mu_{t}=\rho_{t}dx$,
since $\rho_{t}(x)=\rho_{0}(\partial\phi_{t}^{*})/MA(\phi_{t}(x))\leq1/(1-t)$.
Hence $H(\mu_{t})=\int\rho_{t}\log\rho_{t}dx\leq-\int\rho_{t}\log(1-t)dx=-\log(1-t)$.\end{proof}
\begin{rem}
The previous lemma appears as Theorem 3.6 in \cite{Carrillo-f-p}
(in the case of model power laws). However, the verification of the
Gamma convergence of the functionals $F_{\beta}(:=E+H/\beta)$ towards
$E$ was not provided in the proof in \cite{Carrillo-f-p}. This convergence
problem is a special case of the following general problem: consider
a measure $\mu_{0}$ on a topological space $X$ (for example $\R^{D})$
and denote by $H_{\mu_{0}}(\mu)$ the entropy of a probability measure
relative to $\mu_{0}.$ Let $E(\mu)$ be defined by a singular integral
operator on $\mathcal{P}(X)$ with lower semi-continuous kernel $W(x,y).$
Showing that $F_{\beta}:=E+H_{\mu_{0}}/\beta$ Gamma converges towards
$E,$ as $\beta\rightarrow\infty,$ appears to be a rather subtle
problem in general and requires an appropriate compatability between
$\mu_{0}$ and the kernel $W(x,y).$ For example, when $X=\R^{2}$
or $X=\R$ and $W(x,y)=-\log|x-y|$ the convergence can be shown to
hold when $\mu_{0}$ is sufficently regular in the potential theoretic
sense (i.e. $\mu_{0}$ has regular asymptotical behaviour in the sense
of \cite{s-t} or satisfies a Bernstein-Markov property in the sense
of \cite{b-b-w}). We shall not develop this further here, but just
point out that the regularity property in question does hold when
$\mu_{0}$ is, for example, Lesbegue measure $dx$ on $\R$ (for non-trivial
reasons). However, in general, it is not enough to assume that $E(\mu_{0})<\infty.$
In general, as in the proof of the previous lemma, in order to verify
the Gamma convergence it is enough to verify the following condition
\begin{itemize}
\item For any $\mu$ such that $E(\mu)<\infty$ there exists a sequence
$\mu_{j}$ converging weakly to $\mu$ such that $\mu_{j}$ is absolutely
continuous with respect to $\mu_{0}$ and satisfies $E(\mu_{j})\rightarrow E(\mu)$ 
\end{itemize}
From the statistical mechanical point of view this condition appears
naturally in the large deviation theory of the $N-$particle Gibbs
measures corresponding to the pair interaction $W(x,y)$ when $\beta_{N}\rightarrow\infty$
\cite[Hypothese 4, page 2373]{c-g-z}. In the model case of power
laws and $\mu_{0}=dx$ on $X=\R^{D}$ the regularizations $\mu_{j}$
above can alternatively be defined using a convolution of $\mu_{0}$,
by exploting the convexity of the corresponding functional $E$ wrt
the usual affine structure (which is a classical, but non-trivial
fact) \cite{c-g-z}.
\end{rem}
With the previous lemma in place we get the following
\begin{prop}
When $\alpha\leq1$ Proposition \ref{prop:weak solution} holds for
$\beta=\infty,$ as well. \end{prop}
\begin{proof}
This follows immediately from combining Prop \ref{prop:weak solution}
with the stability property of Lemma \ref{lem:stability-wrt-beta}
using that, in this particular case, $W'(x-y)(\psi(x)-\psi(y)$ defines
a continuous function on $\R^{2},$ since $\psi$ is assumed smooth. \end{proof}
\begin{rem}
Recall that for any finite measure $\mu$ on $\R$ the limit 
\begin{equation}
H_{\mu}(x):=\lim_{\epsilon\rightarrow0}\int_{|x-y|\geq\epsilon}\frac{\mu(y)}{x-y}\label{eq:def of hilb transform using trunc}
\end{equation}
exists for a.e. $x$ in $\R$ \cite[Theorem, page 1085]{lo} and defines
the Hilbert transform $H_{\mu}(x)$ of $\mu.$ Formally, the weak
McKean-Vlasov equation for a repulsive logarithmic interaction $(\alpha=1)$
is equivalent to a weak solution to the continuity equation with drift
$v_{t}=H_{\mu}.$ However, in general $H_{\mu}$ is not in $L_{loc}^{1}(\R),$
even if $f\in L^{1}(\R)$ and hence some further a priori regularity
on $\mu_{t}$ is needed in order to even make sense of the corresponding
continuity equation (which requires that $H_{\mu}\in L^{1}(\mu));$
see the discussion in \cite{fo0}). For example if $\mu_{t}$ is in
$L^{p}(\R)$ for some $p>1$ then so is $H_{\mu}$ by Riesz classical
theorem.
\end{rem}
In the logarithmic case the uniqueness of weak solutions $\mu_{t}$
to the McKean-Vlasov equation was established in \cite{c-l} for $V$
quadratic and in \cite{c-d-g,fon} under the assumption that the Fourier
transform of $V$ has exponential decay (in particular, $V$ is real-analytic).
The proofs in \cite{c-l} exploit that the Hilbert transform $H_{\mu}(x)$
is, for almost any $x,$ the boundary value along the real axes of
the Cauchy transform $G_{\mu}(z)$ of $\mu,$ which defines a holomorphic
function on the upper half plane. Then a weak solution in the sense
of the previous proposition corresponds to a strong solution for a
complex Burger type equation in the upper half plane. However, the
uniqueness of weak solutions seems to be open when $V$ is only assumed
smooth and $\lambda-$convex (in \cite{fon} the uniqueness in the
class of solutions $\mu_{t}$ in $L^{p}(\R)$ for $p\geq2$ is established).
Moreover, in the case of attractive power-laws with $\alpha=0$ uniqueness
of weak solutions fails, in general (see Section \ref{sub:Convex-interactions-in}). 

Let us also recall that in the model case of repulsive power-law with
exponent $\alpha\in]0,2[$ the corresponding drift $V_{t}$ can, at
least formally, be written as minus the gradient of the fractional
Laplacian $(-\Delta)^{-(1-\alpha/2)}\rho_{t}.$ Certain weak solutions
of the corresponding evolutions equations are constructed in \cite{b-k-m}
when $D=1$ and in \cite{c-v} when $D\geq1,$ using an elaborate
regularization scheme involving several parameters. However, a key
point of the Wassestein gradient flow approach is that the constructed
limiting curve $\mu_{t}$ is an EVI gradient flow and thus automatically
uniquely determined. 
\begin{rem}
\label{rem:incorrect}Let now $W$ be a repulsive power law with $\alpha\in[1,2[$
(and take, for example, $V(x)=Cx^{2}$$).$ In \cite[Remark 2.2]{Carrillo-f-p}
it is claimed that $E_{W,V}(\mu)<\infty$ implies that $\mu$ is absolutely
continuous wrt $dx,$ for $t>0,$ and as a consequence it is claimed
in Remark \cite[Remark 3.10]{Carrillo-f-p} that $\mu_{\beta}(t)$
is a weak solution in the sense of Prop \ref{prop:weak solution}
also for $\beta=\infty$ (since $E_{W,V}<\infty$ along the EVI gradient
flow). But the first claim appears to be incorrect. Indeed, in the
logarithmic case it is well-known that there are measures $\mu$ which
are not absolutely continuous wrt $dx,$ but with the property $E_{W,V}(\mu)<\infty,$
for example measures with sufficently small Haussdorf dimension (such
as the standard Cantor set). A similar counter example applies when
$\alpha>1$ \cite{Bea}. On the other hand, it may very well be that
using further properties of the EVI gradient flow (for example, that
the metric slopes are finite for $t>0$) one can establish the first
claim, or directly the second claim, for any repulsive powerlaw. Alternatively,
it seems likely that, assuming that $\mu_{0}\in L^{p}(\R^{D})$ for
some $p>1$ one can show that $\mu_{t}\in L^{p}(\R^{D}),$ by a regularization
procedure, as shown in \cite{c-v} in a related context (when $\alpha\leq1$
this follows from the viscosity approach in \cite{b-k-m}, when $V=0)$.
But we will not go further into this here.
\end{rem}

\section{\label{sec:Further-examples}Further applications}

In this section we give some further examples illustrating the general
structure of the applications discussed in Section \ref{sec:Singular-pair-interactions}.

\subsection{Variants of strongly singular power-laws in 1D}

The next example generalizes the power-law in Section \ref{sec:Singular-pair-interactions}
to pair interaction which are not translationally invariant:
\begin{example}
\label{exa:spin}Given a function $g(x,y)$ which is positive, bounded
and with bounded first and second partial deriviatives consider the
following pair interaction potential:
\begin{equation}
W(x,y):=\frac{g(x,y)}{|x-y|^{s}}.\label{eq:W in terms of g}
\end{equation}
 and let $E^{(N)}(x_{1},x_{2},....,x_{N})$ be the corresponding pair
interaction energy. Then the assumptions in Section \ref{sec:A-general-convergence}
are satisfied, as follows essentially by the same arguments as in
Section \ref{sec:Singular-pair-interactions}. We recall that such
interactions appear, for example. as spin type Hamiltonians in the
mathematical physics litteratyre. To briefly explain this consider
$N$ particles located at $x_{1},...,x_{N}$ with internal spin type
degrees of freedom $S_{1},...,S_{N}$ taking values in the unit-sphere
in $\R^{n}.$ The corresponding spin type Hamiltonian may be defined
by 
\[
H(x_{1},...,x_{N};S_{1},...,S_{N}):=\sum_{i,j}J_{x_{i}x_{j}}(1-S_{i}\cdot S_{j}),
\]
 for a given function $J_{xy}$ on $\R^{2},$ the spin interaction.
These models have been studied extensively in lattice models where
the positions $x_{1},...x_{N}$ are fixed once and for all and put
on a lattice, for example $x_{i}=i$ in the 1D setting. The case $n=1$
and $J_{xy}=|x-y|^{-\alpha}$ is then usually called the $\alpha-$Ising
model (while the general case $n\geq1$ is called the $O(n)-$model);
see\cite[Section 4]{bbdr} for the corresponding static mean field
limit and \cite{g-l} for a similar stochastic dynamical lattice model
with long range interactions, where a McKean-Vlasov type equation
appears in the macroscopic limit. The present setting applies to the
opposite setting where the positions $x_{1},...x_{N}$ are free while
$S_{i}=S(x_{i})$ for a given smooth map $S$ from $\R$ into the
unit-sphere in $\R^{n},$for $n\geq2.$ One then sets $g(x,y)=(1-S(x)\cdot S(y))$
in formula \ref{eq:W in terms of g}. For example, this happens when
the system is coupled to a strong exterior magnetic field $B(x),$
effectively forcing $S(x)$ to be parallell to $B.$ 
\end{example}
Let us also give some 1D examples which are not pair interactions:
\begin{example}
Consider the following 1D $N-$particle interaction energy: 
\[
E^{(N)}(x_{1},x_{2},....,x_{N}):=\frac{1}{N^{3}}\sum_{i,j,k}\frac{1}{\left(|x_{i}-x_{j}|+|x_{k}-x_{i}|+|x_{k}-x_{j}|\right)^{s}}\,\,\,\,s\in]0,1[
\]
where the sum funs over all indices $i,j,k$ in $[1,N]$ such that
$i,j$ and $k$ do not all coincide. This means that $E^{(N)}(x_{1},x_{2},....,x_{N})$
blows up precisely when at least three points merge (but remains bounded
if only two points merge). If $x_{i}<x_{j}<x_{k}$ then the function
appearing above is clearly convex and in $L_{loc}^{1}$ and hence
the assumptions in Section \ref{sec:A-general-convergence} are satisfied,
as discussed in Section \ref{sec:Singular-pair-interactions}. 
\end{example}

\subsection{\label{sub:Convex-interactions-in}Globally quasi-convex interactions
in any dimension}

Next, we give some extensions from the case $D=1$ to higher dimensions. 
\begin{thm}
Let $w_{m},$ $m=1,...,M$ be quasi-convex funtions on $(\R^{D})^{m}$
and $V(x)$ a quasi-convex function on $\R^{D}.$ Assume that at the
initial time $t=0$ 
\[
\lim_{N\rightarrow\infty}(\delta_{N})_{*}\mu_{t}^{(N)}=\delta_{\mu_{0}}
\]
in the $L^{2}-$Wasserstein metric. Then, at any positive time 
\[
\lim_{N\rightarrow\infty}(\delta_{N})_{*}\mu_{t}^{(N)}=\delta_{\mu_{t}}
\]
in the $L^{2}-$Wasserstein metric, where $\mu_{t}$ is the EVI-gradient
flow on $\mathcal{P}_{2}(\R)$ of the corresponding free energy functional
$F_{\beta},$ emanating from $\mu_{0}.$ \end{thm}
\begin{proof}
This is proved essentially as in the proof Theorem \ref{thm:w}. In
fact, the proof is even simpler by the assumtion of global semi-convexity.
One can then approximate each $w_{m}$ on all of $\R^{D}$ with a
sequence $w_{m}^{(R)}$ of quasi-convex functions (with the same $\lambda)$
increasing to $w_{m}$ and such that $|w_{R}(x)|=o(|x|^{2})$ as $|x|\rightarrow\infty,$
for in fixed $R.$ Indeed, by the quasi-convexity assumption it is
enough to treat the case when $w_{m}(=\phi)$ is convex. In that case
we can, for example, take $\phi^{(R)}:=\inf_{y\in\R^{Dm}}R|x-y|+\phi(y)).$
It is then shown, precisely as in the proof of Theorem \ref{thm:w},
that the corresponding $N-$partical energy $E^{(N)}(x_{1},...x_{N})$
satisfies the assumptions in Section \ref{sub:The-assumptions-on}. 
\end{proof}
In particular the previous theorem applies to the non-smooth weakly
singular attractive pair interactions 
\begin{equation}
W(x,y):=w(|x-y|)=|x-y|^{1+\alpha},\,\,\,0\leq\alpha<1.\label{eq:attractive pair int}
\end{equation}
In the purely stochastic setting such pair interactions appear, for
example, in the study of granular media and swarming models (see \cite{c-d-a-k-s}
and references therein). The purely determinstic setting for such
pair interaction has been studied in depth in \cite{c-d-a-k-s}. For
example it is shown that, given any initial $\mu_{0}$ with compact
support the corresponding EVI gradient flow $\mu_{t}$ of $E$ is
a Dirac mass, i.e. $\mu_{t}=\delta_{x(t)},$ when $t\geq T$ for some
finite tume $T.$ But as far as we know there are no results concerning
such weakly singular attractive interactions in the purely stochastic
setting $(\beta_{N}<\infty)$, apart from the case when $\alpha=1$
and $D=1$ where the theory of scalar conservation laws applies \cite{b-t}
(the case of smooth potentials with polynomial growth is studied in
\cite{g=0000E4,mal}).
\begin{rem}
For the attractive pair interactions of the form \ref{eq:attractive pair int}
it is shown in \cite{c-d-a-k-s} that the corresponding minimal subdifferental
is given by 
\[
(\partial^{0}\mathcal{W})(\mu):=\int_{\{x\neq y\}}(\nabla_{x}W)(x,y)\mu(y),
\]
 which is a point-wise well-defined function. Setting $(\nabla w)(0):=0$
(which, by symmetry, coincides with the minimal subdifferential of
$w(x)$ at $0)$ this means that $(\partial^{0}\mathcal{W})(\mu)=\mu*\nabla w.$
\end{rem}
By Theorem \ref{thm:evi is subdiff gradient flow}, the macroscopic
limit $\mu_{t}$ is a weak solution of the corresponding equation
\[
\frac{d}{dt}\mu_{t}=\frac{1}{\beta}\Delta\mu_{t}+\nabla\cdot(\mu_{t}(\partial^{0}\mathcal{W})(\mu_{t}))
\]
When $\beta=\infty$ it is well-known that a weak solution is not
unique, in general. For example, when $D=1$ and $\alpha=0$ the map
$\mu\mapsto\mu*\nabla w$ gives a correspondence between weak solutions
$\mu_{t}$ as above and weak solution $u_{t}$ of the scalar conservation
law (Burger's equation)

\[
\partial_{t}u_{t}=\frac{1}{\beta}\partial_{x}^{2}u_{t}+\partial_{x}(u_{t}^{2})/2
\]
satisfying $|u|\leq1,\,\partial_{x}u\geq0.$ By the general theory
of scalar conservation laws a weak solution $u_{t}$ is uniquely determined
if it is an entropy solution. Moreover, $u_{t}$ is an entropy solution
iff $\mu_{t}$ is the EVI gradient flow of $E(\mu)$ (as follows,
for example, from the stability of entropy solutions and EVI gradient
flows when $\beta\rightarrow\infty;$ see \cite{b-c-f} for futher
results).

\section{\label{sub:Comparison-with-stability}Relations to stability properties
of gradient flows in the deterministic setting}

Let us start by stressing that, in the completely deterministic setting
(i.e. Setting $1$ in the introduction of the paper) the theory of
Wasserstein gradient flows has certainly been used before to establish
mean field limits going beyond the classical setting when $F$ is
locally Lipchitz continuous. Indeed, as explained in \cite{c-d-a-k-s},
as soon as the energy functional $E(\mu)$ is lsc and has the property
that 
\begin{itemize}
\item $E(\mu)$ is $\lambda-$convex along generalized geodesics in the
$L^{2}-$Wasserstein space $\mathcal{P}_{2}(\R^{D})$ 
\item The gradient flow of $E(\mu)$ preserves particles, i.e. it preserves
discrete measures of the form \ref{eq:empirical measure} 
\end{itemize}
then the existence of a mean field limit follows directly from the
$\lambda-$contractivity of the gradient flow of $E$ on $\mathcal{P}_{2}(\R^{D}).$
In particular, as shown in \cite{c-d-a-k-s}, these assumption are
satisfied $w$ is convex on all of $\R^{D}$ and even (and in particular
locally Lip continuous), for example when $W(x)=|x|^{\text{1+\ensuremath{\alpha}}},$
for $\alpha\geq0.$ Such a potential is always attractive. However,
even if the pair interaction potential $w$ is Lipschitz continuous,
i.e. $F$ is bounded, the property of preservation of particles fails
when $F$ is repulsive. Moreover, in the strongly singular case the
gradient flow of $E$ never preserves particles (since $E(\mu)=\infty,$
when $\mu$ is discrete). In our approch this problem is bypassed
by instead working with the $N-$particle mean energy functional $\mathcal{E}^{(N)}$
on the Wasserstein space $\mathcal{P}_{2}(\R^{D}).$ On the other
hand in the purely deterministic setting, there is also an alternative
approach using the following stability result in \cite[Thm 11.1.2]{a-g-s}
(see also \cite{a-g-z}) for gradient flows on $\mathcal{P}_{2}(X)$
for $X$ the Euclidean space $\R^{D}$ (or more generally a Hilbert
space).
\begin{thm}
\label{thm:-stab of gradient flo}\cite{a-g-s} (Stability) Suppose
that $\Phi_{N}$ and $\Phi$ are functionals on $\mathcal{P}_{2}(X)$
which are $\lambda-$convex along generalizes geodesics and such that$\Phi_{N}$
strongly $\Gamma-$converges to $\Phi$ (see Definition \ref{def:Gamma conv})
and $\Phi_{N}$ is uniformly coercive. Let $\mu_{N}(t)$ and $\mu(t)$
be the corresponding EVI gradient flows in $\mathcal{P}_{2}(X)$ emanating
from $\mu_{N,0}$ and $\mu_{0},$ respectively. If $\mu_{N,0}\rightarrow\mu$
in $\mathcal{P}_{2}(X),$ as $N\rightarrow\infty$ and $\limsup_{N\rightarrow\infty}\Phi_{N}(\mu_{N,0})<\infty,$
then $\mu_{N,0}(t)\rightarrow\mu(t)$ in $\mathcal{P}_{2}(X)$ for
any positive time $t.$\end{thm}
\begin{rem}
\label{rem:stability of gradient flows}If $\Phi(\mu)<\infty$ one
can dispense with the assumption $\limsup_{N\rightarrow\infty}\Phi_{N}(\mu_{N,0})<\infty$
in the previous theorem, using the definition of $\Gamma-$convergence
together with the contractivity property, as in the Step 2 in the
proof of Theorem \ref{thm:general}. Moreover, when $D=1$ it is enough
to assume that $\Phi_{N}$ and $\Phi$ are $\lambda-$convex along
ordinary geodesics in $\mathcal{P}_{2}(\R),$ using that the latter
space is Euclidean.
\end{rem}
Now set $D=1$ and consider a sequence of symmetric, i.e. $S_{N}-$invariant,
lsc functions $E^{(N)}$ on the $N-$particle space $\R^{N}$ which
are $\lambda-$convex on the fundamental domain $\{x_{1}<x_{2}<....<x_{N}\}$
of the $S_{N}-$action, as in Section \ref{sec:Singular-pair-interactions}.
Using the embedding $\delta_{N}$ of $\R^{N}/S_{N}$ in $\mathcal{P}_{2}(\R)$
we can identify $E^{(N)}$ with a sequence of functionals on $\mathcal{P}_{2}(\R)$
set to be equal to $\infty$ on the complement of $\delta_{N}(\R^{N}/S_{N}).$
As observed in Section \ref{sub:An-alternative-approach} $E^{(N)}$
is convex on the quotient space $\R^{N}/S^{N}$ and hence, since $\delta_{N}$
is an isometry and its image is geodesically closed, the corresponding
function on $\mathcal{P}_{2}(\R)$ is also $\lambda-$convex. In order
to apply the previous theorem to pair interactions one can then invoke
the following result from \cite[Prop 2.8, Remark 2.19]{se} (similar
results are used to establish large deviations of the corresponding
Gibbs measures; see \cite{a-g-z} and references therein):
\begin{prop}
\label{prop:gamma conv}\cite{se}Let $E^{(N)}$ be the $N-$point
interaction energy on $(\R^{D})^{N}$ associated to a translationally
invariant radial pair interaction $W(x,y)(:=w(|x-y|)$ such that $W\in L^{1}(\R^{2D})$
and $w(x)$ is monotone in the radial direction and positive when
$|x|\leq1.$ Then the corresponding functionals on $\mathcal{P}_{2}(\R^{D})$
$\Gamma-$converge to $E_{W}(\mu)$.
\end{prop}
Asuming that the initial measure $\mu_{0}$ has the property that
$E(\mu_{0})<\infty$ Theorem \ref{thm:-stab of gradient flo} combined
with the previous proposition thus implies the existence of a mean
field limit $\mu_{t}$ of the corresponding determinstic systems,
as in Section \ref{thm:deterministic strong}. But the advantage of
our general convergence results in Theorem \ref{thm:general}, when
applied to the purely deterministic setting, is that it allows $E(\mu_{0})=\infty$
and moreover there is no need to establish the $\Gamma-$convergence
of the interaction energies $E^{(N)}$. Indeed, the convergence assumption
1 in Section \ref{sub:The-assumptions-on} for the corresponding mean
energy functional $\mathcal{E}^{(N)}$ on $\mathcal{P}_{2}(\R^{D})$
is almost trivially satisfied for any pair interaction (or more generally,
for any $m-$point interaction).

\end{document}